\documentclass[letterpaper, DIV=13]{scrartcl}

\usepackage{amssymb, amsthm, mathtools,bbm}
\usepackage[square,sort,comma,numbers]{natbib}
\usepackage{enumitem}

\usepackage[none]{hyphenat}
\newtheorem{theorem}{Theorem}

\newtheorem{lemma}[theorem]{Lemma}

\newtheorem{proposition}[theorem]{Proposition}

\numberwithin{equation}{section}
\numberwithin{theorem}{section}

\newcommand{\rr}{{\mathbb{R}}}

\newcommand{\nn}{{\mathbb{N}}}

\newcommand{\tr}{{\operatorname{Tr}\,}}

\newcommand{\beq}[1]{\begin{equation} \label{#1}}
	\newcommand{\eeq}{\end{equation}}

\renewcommand{\epsilon}{\varepsilon}

\newcommand{\spa}{\operatorname{span}}

\def\be{\begin{equation}}
	\def\ee{\end{equation}} 
\DeclareMathOperator{\spec}{spec}
\DeclareMathOperator{\dist}{dist}
\DeclareMathOperator{\gap}{gap}
\DeclareMathOperator{\ran}{ran}

\begin{document}
	
\title{ \raggedright  The Spectral Gap and Low-Energy Spectrum\\ in 
	Mean-Field Quantum Spin Systems}

	\author{Chokri Manai and Simone Warzel\\
	{\small Department of Mathematics}\\[-1ex]
	{\small Technische Universit\"at M\"unchen, Germany}}
	\date{\vspace{-.3in}}
	
	\maketitle

	\minisec{Abstract}
	A semiclassical analysis based on spin-coherent states is used to establish a classification and novel simple formulae for the spectral gap of mean-field spin Hamiltonians. For gapped systems we provide a full description of the low-energy spectra based on a  second-order approximation  to the semiclassical Hamiltonian hence justifying fluctuation theory at zero temperature for this case.  
	We also point out a shift caused by the spherical geometry in these second-order approximations. 
	
	
	\bigskip
	
\tableofcontents
\section{Introduction}

Mean-field quantum spin systems are ubiquitous in effective descriptions of a variety of phenomena. A popular example is the family of  Lipkin-Meshkov-Glick Hamiltonians, which were originally conceived in~\cite{LMG65-1,LMG65-2,LMG65-3}  to explain shape transitions in nuclei, but also feature in descriptions of Bosons in a double well and quantum-spin tunnelling in molecular magnets~\cite{Cirac:1998aa,MMBook14}. 
This family includes the 
quantum Curie-Weiss Hamiltonian, whose simplicity continues to draw the attention of many communities~\cite{BMT66,Ribeiro:2008aa,Chayes:2008aa,Landsman:2020aa,VGR20,Bjornberg:2020aa,Bulekov:2021aa}. In particular, in~\cite{Bapst_2012} such models were used to test conjectures related to quantum annealing for which information about the spectral gap is crucial. 

Most mean-field spin Hamiltonians in the literature are defined in terms of a non-commuting self-adjoint polynomial of the three components of the total spin-vector
$	 \mathbf{S} = \sum_{n=1}^N  \mathbf{S}(n) $.  
For a system of $ N $ interacting qubits,  the Hilbert space on which these operators act is the tensor product $ \mathcal{H}_N = \bigotimes_{n=1}^N \mathbb{C}^2 $.  
The vectors $ \mathbf{S}(n) = \mathbbm{1} \otimes \dots \otimes \ \mathbf{s} \ \otimes \dots \otimes \mathbbm{1}  $ stand for the natural lift of 
the spin vectors $  \mathbf{s}  = (s_x, s_y, s_y ) $ to the $ n $-th component of the tensor product. On each copy of $ \mathbb{C}^2 $, the spin vector coincides with the three generators  of $ SU(2) $:
\begin{equation*}
s_x = \frac{1}{2} \left( \begin{matrix} 0 & 1 \\ 1 & 0 \end{matrix}\right) , \quad s_y = \frac{1}{2} \left( \begin{matrix} 0 & -i \\ i & 0 \end{matrix}\right) , \quad s_z = \frac{1}{2} \left( \begin{matrix} 1 & 0 \\ 0 & -1 \end{matrix}\right) . 
\end{equation*}
By definition, a non-commuting self-adjoint polynomial $  \textrm{P}\Big(\tfrac{2}{N}   \mathbf{S} \Big) $ is a  finite linear combination of products of the  three rescaled components of the total spin-vector in which each product is averaged (corresponding to Weyl ordering) under the three components so that  it becomes a self-adjoint operator. This renders the associated mean-field Hamiltonian
\begin{equation}\label{eq:Ham}
	H = N \ \textrm{P}\Big(\tfrac{2}{N}   \mathbf{S} \Big) .
\end{equation}
self-adjoint on $ \mathcal{H}_N $. The dependence on the particle number is two-fold and quintessential for the mean-field nature. 
The scaling of the spin ensures that the operator is norm-bounded by one, $ \left\| \tfrac{2}{N}  S_\xi \right\| \leq 1 $ for all $ \xi= x,y, z $. Moreover, the prefactor $ N $ forces the energy $ H $  to be extensive.  

For example, the Lipkin-Meshkov-Glick  model is given by
$ P( \mathbf{m}) = -\alpha m_y^2 - \beta m_z^2 -\gamma m_x $ with $ \alpha,\beta , \gamma \in \mathbb{R} $. The special case $ \alpha = 0 $, $ \beta= 1 $ corresponds to the quantum Curie-Weiss model with $ \gamma $ playing the role of the transversal,  external magnetic field. In~\cite{Bapst_2012}, the $p$-spin generalization of the Curie-Weiss model has been considered, for which  $  P( \mathbf{m}) = - \beta m_z^p -\gamma m_x $ with $ p \in \mathbb{N} $. In these examples, the monomials do not involve mixed products, which makes the Weyl-ordering obsolete. To illustrate this issue, consider $ P( \mathbf{m}) = m_x m_y $, which then results in $ \textrm{P}\big(2   \mathbf{S} /N \big) =2 \big(  \mathbf{S}_x  \mathbf{S}_y +  \mathbf{S}_y  \mathbf{S}_x \big) / N^2$.  

Since $ H $ in~\eqref{eq:Ham} is a function of the total spin $  \mathbf{S}  $, it is block diagonal with respect to the decomposition of the  tensor-product Hilbert space according to the irreducible representations of the total spin corresponding to the eigenspaces of  $ \mathbf{S}^2 $ with eigenvalues $ J(J+1) $, i.e.
\begin{equation}\label{def:blockH}
\mathcal{H}_N \equiv  \bigoplus_{J=\frac{N}{2} - \lfloor \frac{N}{2}\rfloor}^{N/2} \bigoplus_{\alpha=1}^{M_{N,J}}  \; \mathbb{C}^{2J+1} , \qquad M_{N,J}  = \frac{2J+1}{N+1} \binom{N+1}{\frac{N}{2} +J + 1 } . 
\end{equation}
The total spin $ J $ of $ N $ qubits can take any value from $ N/2 $ down in integers to either $ 1/2 $ if $ N $ is odd or $ 0 $ if $ N $ is even. The degeneracy of the representation of spin $ J $ in this decomposition is $ M_{N,J}  $ \cite{Mihailov:1977aa}. 
On each block $ (J, \alpha) $, the Hamiltonian~\eqref{eq:Ham} then acts as the given polynomial of the generators of the irreducible representation of $ SU(2) $ on $ \mathbb{C}^{2J+1} $.

The analysis of such systems in the limit of large spin quantum number $ J$ is known \cite{Lieb73,HeppLieb73,MN04,MN05,Biskup:2007aa} to be facilitated by 
Bloch coherent states on the Hilbert space $  \mathbb{C}^{2J+1}  $. They are parametrized by an angle $\Omega =(\theta, \varphi) $ on the unit sphere $ S^2 $ with $ 0\leq \theta \leq \pi $,\; $ 0 \leq \varphi \leq 2 \pi $. In bra-ket-notation, which we will use in this paper, the Bloch-coherent states are given by
\begin{equation}\label{def:cs}
\big| \Omega, J \rangle := U(\theta,\varphi) \  \big| J \rangle , \qquad U(\varphi,\theta) :=  \exp\left( \frac{\theta}{2} \left( e^{i\varphi} S_- - e^{-i\varphi} S_+ \right) \right)  
\end{equation} 
The reference vector  $  \big| J \rangle \in  \mathbb{C}^{2J+1}  $ is the normalized eigenstate of the $ z $-component of the spin corresponding to (maximal) eigenvalue $ J $ on the Hilbert space $ \mathbb{C}^{2J+1} $. The operators $ S_\pm = S_x \pm i S_y $ are the spin raising and lowering operators of the irreducible representation of $ SU(2) $ on $  \mathbb{C}^{2J+1} $. 

Bloch coherent states have many remarkable features. First and foremost, they form an overcomplete set of vectors as expressed through the resolution of unity on  $  \mathbb{C}^{2J+1} $:
\begin{equation}\label{eq:completeness}
 \frac{2J+1}{4\pi} \int \big| \Omega, J \rangle \langle \Omega, J \big|  \ d\Omega  = \mathbbm{1}_{  \mathbb{C}^{2J+1}} .
\end{equation}
Every linear operator $  G $ on $  \mathbb{C}^{2J+1} $ is associated with a lower and upper symbol. The lower symbol is  $ G(\Omega,J) :=  \langle \Omega , J \big|  G  \big| \Omega, J \rangle $, and the upper symbol is characterized  through the property
\begin{equation}
 G = \frac{2J+1}{4\pi} \int  g(\Omega,J) \,  \big| \Omega, J \rangle \langle \Omega, J \big| \, d\Omega  .
 \end{equation}
The choice of $ g $ is not unique. E.g.\ through an explicit expression~\cite{Kutzner:1973rz}, one sees that there is always an arbitrarily often differentiable choice,  $ g(\cdot, J) \in C^\infty(S^2)  $.  
More properties of coherent states are collected in Appendix~\ref{app:cs}; see also~\cite{ACGT72,Per86,Gaz09,CR12}.

\subsection{Semiclassics for the free energy} 
The lower and upper symbol feature prominently in Berezin and Lieb's semiclassical bounds~\cite{Ber72,Lieb73,Simon:1980} on the partition function associated with a self-adjoint Hamiltonian $ G $ on $ \mathbb{C}^{2J+1} $:
\begin{equation}\label{eq:BerezinLieb}
 \frac{2J+1}{4\pi} \int e^{-\beta G(\Omega,J)} d\Omega \leq  \tr_{\mathbb{C}^{2J+1}}  e^{-\beta G } \leq \frac{2J+1}{4\pi} \int e^{-\beta g(\Omega,J)} d\Omega  . 
 \end{equation}
 In the semiclassical limit of large spin quantum number $ J $, these bounds are known to asymptotically coincide \cite{Lieb73,Simon:1980,Duffield:1990aa}. 
In the same spirit, for any polynomial of the spin operator as in~\eqref{eq:Ham} restricted to $  \mathbb{C}^{2J+1} $ both the upper and lower symbols agree to leading order in $ N $ with the corresponding classical polynomial function on the unit ball $ B_1 $, which parametrises the Hilbert space~\eqref{def:blockH} semiclassically. Using spherical coordinates $ \mathbf{e}(\Omega) = \left(\sin\theta \cos\varphi ,  \sin\theta \sin\varphi , \cos\theta \right) \in S^2 $, one has 
\begin{equation}\label{eq:uppersymbclasss}
 \sup_{0 \leq J \leq N/2} \left\|  \textrm{P}\Big(\tfrac{2}{N}  \mathbf{S} \Big)\Big|_{  \mathbb{C}^{2J+1}} -  \frac{2J+1}{4\pi} \int    \textrm{P}\Big(\tfrac{2J}{N} \mathbf{e}(\Omega)  \Big)  \big| \Omega, J \rangle \langle \Omega, J \big| \, d\Omega  \right\| \leq \mathcal{O}(N^{-1}) 
\end{equation}
for the operator norm $ \| \cdot \| $ on $  \mathbb{C}^{2J+1} $. We use the Landau $ \mathcal{O} $-notation, i.e., the error on the right is bounded by $ C N^{-1} $ with a constant $ C $ which only depends on the coefficients of the polynomial. This statement is a quantitative version of Duffield's theorem~\cite{Duffield:1990aa}. Since it is hard to locate general quantitative error estimates, we include a proof as Proposition~\ref{prop:Druff1} in the Appendix. 
As is recalled in Proposition~\ref{prop:Druff2},  the lower symbol then shares the same classical asymptotics
$$  \sup_{0 \leq J \leq N/2} \sup_{\Omega} \left| \langle \Omega , J \big|   \textrm{P}\big(\tfrac{2}{N}  \mathbf{S} \big)\Big|_{  \mathbb{C}^{2J+1}}  \big| \Omega, J \rangle -    \textrm{P}\Big(\tfrac{2J}{N} \mathbf{e}(\Omega)  \Big)  \right| \leq  \mathcal{O}(N^{-1})  .
$$
The Berezin-Lieb inequalities~\eqref{eq:BerezinLieb} immediately imply that the free energy of~\eqref{eq:Ham} is determined by minimizing a variational functional involving the classical energy $  \textrm{P} $ on the unit ball and the (shifted) binary entropy 
$$ I(r)  := \begin{cases} - \frac{1+r}{2} \ln\frac{1+r}{2} - \frac{1-r}{2} \ln  \frac{1-r}{2} , & \qquad r \in (0,1) .\\
\qquad\qquad 0 , & \qquad r = 1, \\
\qquad\quad\ln 2 , &  \qquad r = 0 . 
\end{cases} $$
A straightforward, rigorous saddle-point evaluation, which we spell out in the proof of  Proposition~\ref{prop:free2} -- a generalisation to mean-field models with regular symbols  -- yields the pressure as a min-max variational principle on $ B_1 $. 
\begin{proposition}
For a mean-field Hamiltonian $ H = N \ \textrm{P}\Big(\tfrac{2}{N}   \mathbf{S} \Big) $ with a non-commuting self-adjoint polynomial $ \textrm{P} $, the pressure for any $ \beta > 0 $ is given by:
\begin{equation}\label{eq:freee}
p(\beta) := \lim_{N\to \infty} N^{-1} \ln \tr \exp\left( -\beta N  \textrm{P}\Big(\tfrac{2}{N}  \mathbf{S} \Big) \right) = \max_{r \in [0,1]} \left\{ I(r) - \beta \min_{\Omega \in S^2}   \textrm{P}\left(r\mathbf{e}(\Omega) \right)  \right\} .
\end{equation}
\end{proposition}
As a special case with constant field $ P(\mathbf{m}) = -\lambda  m_z  $,  $ \lambda \geq 0 $, one obtains the Legendre relation 
$$
\ln 2 \cosh\left( \beta \lambda \right)= \max_{r \in [0,1]} \left[ I(r) +  \beta \lambda r\right] .
$$
By inverting this Legendre transform, $ I(r) =   \min_{\lambda \geq 0 }  \left[ \ln 2 \cosh\left( \beta \lambda \right) -  \beta \lambda r\right]  $, one may rewrite~\eqref{eq:freee} in the slightly more familiar form
$$
p(\beta)  =  \max_{r \in [0,1]}  \min_{\lambda \geq 0 } \left\{\ln 2 \cosh\left(\beta \lambda \right) -   \beta \left( \min_{\Omega \in S^2}   \textrm{P}\left(r\mathbf{e}(\Omega) \right) + \lambda r \right)   \right\} .
$$

Investigations of the free energy have been of great interests over many decades~\cite{BMT66,Lieb73,FSV80,Cega:1988aa,RagWer89,Duffield:1990aa,Chayes:2008aa,Ven:2020aa}. They are however not the main focus of this paper. We therefore conclude this topic with only one brief comment on the literature.

Among the numerous results, it is worth mentioning that alternatively to the sketched approach via the Berezin-Lieb inequalities, the formula~\eqref{eq:freee} may be derived exploiting the exchange symmetry of the mean-field Hamiltonian using a version of the quantum de Finetti theorem.  This road was taken by Fannes, Spohn and Verbeure~\cite{FSV80}. 
Their result essentially covers all Hamiltonians $ H = \sum_{p=1}^m A^{(p)} $ on $  \mathcal{H}_N  $ with exchange-symmetric $ p $-spin interactions $ A^{(p)} $ and yields 
$$
p(\beta) = \sup_\varrho \left[ E(\varrho) - \beta \sum_{p=1}^m \tr A^{(p)} \varrho^{\otimes p}    \right] .
$$
For a mean-field system of $  N $ spin-$1/2 $, 
the supremum is over states which are parametrized by the Bloch sphere, $ \varrho = \frac{1}{2} \mathbbm{1}_{\mathbb{C}^2}  - r\mathbf{e}(\Omega) \cdot   \mathbf{S} $, and whose entropy is $ E(\varrho) = - \tr \varrho \ln \varrho = I(r) $.  
As an aside, we  note that this powerful de-Finetti-approach generalizes from spin-$1/2 $ to spin-$ s $~\cite{FSV80}, and then covers the results on the free energy obtained in \cite{Bjo16,BjRR22}. E.g.\ for the exchange Hamiltonian $ T (\psi \otimes \phi) \coloneqq \phi \otimes \psi$ on $\mathbbm{C}^{d}$ for which
$\tr T  \varrho \otimes \varrho = \sum_{i,j=1}^d \lambda_i \lambda_j \tr T |u_i \rangle \langle u_i | \otimes |u_j \rangle \langle u_j | = \sum_{i=1}^d \lambda_i^2 $, one immediately gets  $p(\beta) = \sup_{\mathbf{\lambda} \in \Delta^{d}} \sum_{i=1}^d(- \lambda_i \ln \lambda_i  -\beta \lambda_i^2 )$, where $ \mathbf{\lambda} = (\lambda_1,\dots , \lambda_d) \in \Delta^{d} $ is the vector of eigenvalues of the state $ \varrho $ on $\mathbbm{C}^{d}$.

\subsection{Spectral gap from semiclassics plus fluctuations}

The main result of this paper is a simple quasi-classical explanation and formulae  for the low-energy part of the spectrum of a self-adjoint mean-field operator $ H $ as in~\eqref{eq:Ham} in the limit $ N \to \infty $. We will denote  by 
$ E_0(H) \leq E_1(H) \leq E_2(H) \leq \dots $ the ordered sequence of  its eigenvalues counted with multiplicities. In particular, the existence and leading asymptotic value of the spectral gap 
$$
\gap H = E_1(H) - E_0(H) 
$$
can be read of from the location of the minimum $ \mathbf{m}_0  $ of the polynomial $ P: \mathbb{R}^3 \to \mathbb{R} $ restricted to the unit ball $ B_1 $.  
In case the minimum is unique and located on the surface $ S^2 $, the operator~\eqref{eq:Ham} generically has a spectral gap. To leading order in $ N $  the value of this gap is in fact completely determined by the coefficients of the quadratic polynomial which is uniquely associated with $ P $. In view of the notorious difficulty of determining the spectral gap in quantum lattice systems \cite{Haldane:1983aa,Affleck:1986aa,Affleck:1989aa,KNS01}, this simplicity might be somewhat surprising. 

Broadly speaking, our results are in accordance with the general belief of fluctuation theory that the second-order approximation to $ P $, which involves the gradient $ \nabla P( \mathbf{m}_0 ) $ and the Hessian $ D_P( \mathbf{m}_0 )= \left( \partial_j \partial_k P( \mathbf{m}_0 ) \right)_{j,k=1}^3 $, yields the description of the low-energy spectra. Related statements have been proven in the context of mean-field Bose systems (see e.g.~\cite{Grech:2013bh,Lewin:2015ay}).
For the precise formulation of such a result for quantum spin systems and in order to point out a subtlety caused by the geometry, we need some basic geometric facts on functions on $ B_1 $. 

If $  \mathbf{m}_0 \in S^2 $ is a minimum of $ P $ on $ B_1 $, the gradient  either vanishes or points towards the center of the ball, $ \nabla P( \mathbf{m}_0 ) = -  | \nabla P( \mathbf{m}_0 ) |  \ \mathbf{m}_0 $. 
The quadratic approximation of the polynomial is then given by $ D_P( \mathbf{m}_0 ) $ projected on the directions perpendicular to $ \mathbf{m}_0 $. In terms of the normalized directional vector 
$$
\mathbf{e}_{\mathbf{m}_0} = \frac{\mathbf{m}_0}{| \mathbf{m}_0| } , \qquad \mbox{we set } \qquad Q_\perp  := \mathbbm{1}_{\mathbb{R}^3} - \mathbf{e}_{\mathbf{m}_0}^T \mathbf{e}_{\mathbf{m}_0}   ,
$$ 
which is understood as a linear projection map on $ \mathbb{R}^3 $. 
Introducing a local chart $ \Phi : \mathbb{R}^2 \to T_{\mathbf{m}_0} S^2 $, the linear map on $ \ran Q_\perp \equiv T_{\mathbf{m}_0} S^2 $ given by
\begin{equation}\label{def:perpHessian}
 D_P^{\perp}(\mathbf{m}_0) :=  Q_\perp D_P( \mathbf{m}_0) Q_\perp + | \nabla P( \mathbf{m}_0 ) | \ Q_\perp 
\end{equation}
is then the quadratic approximation to $ P \circ \Phi $ at $ \mathbf{m}_0 $. The shift of the Hessian  in cartesian coordinates by the norm of the gradient $|\nabla P( \mathbf{m}_0 ) |$ is thus an effect of the constraint due to the spherical geometry.

\begin{theorem}\label{thm:pol}
Let $ H $ be a self-adjoint operator on $ \mathcal{H}_N $ of the form~\eqref{eq:Ham} with a non-commuting self-adjoint polynomial $ P: \mathbb{R}^3 \to \mathbb{R} $. Suppose that the minimum of $ P $ restricted to the unit ball $ B_1 $ is unique and located at a point $ \mathbf{m}_0 \in S^2 $ on the unit sphere. Then,
\begin{equation}\label{eq:gappol}
\gap H =   2 \min \left\{ |\nabla P( \mathbf{m}_0 ) |, \sqrt{\det   D_P^{\perp}(\mathbf{m}_0)} \right\}+o(1).
\end{equation}
is the spectral gap above the unique ground state in case the rhs is strictly positive. 
\end{theorem}

This theorem is a special case of Theorem~\ref{thm:unimin}, which deals with mean-field spin Hamiltonians with more general regular  symbols than just polynomials. Theorem~\ref{thm:unimin} also has the corresponding description of the leading asymptotics of the entire low-energy spectrum in terms of fluctuation operators in directions $ \ran Q_\perp $. By Theorem~\ref{thm:polass} these more detailed results also apply to the polynomial case. 
To the best of our knowledge, all these results are entirely new in the mathematics literature. Previously, the spectrum $ \spec H $ of mean-field operators as in~\eqref{eq:Ham} has been identified only on order $ N $. In~\cite[Cor.~2.2]{Ven:2020aa}  it is shown that 
\begin{equation}\label{eq:VV20}
\lim_{N\to \infty}  \dist\left( N^{-1} \spec H , \ran P \right) = 0 ,
\end{equation}
where $ \ran P $ is the range of the function $ P: B_1 \to \mathbb{R} $. \\

To demonstrate the applicability of Theorem~\ref{thm:pol}, we consider the quantum Curie-Weiss  Hamiltonian $H = -\frac{4}{N} S_x^2 - 2 \gamma S_z$, which corresponds to the choice $P(\mathbf{m})  =  -m_x^2-\gamma m_z$. 
		The gradient and Hessian in cartesian coordinates are given by 
		\[\nabla P(\mathbf{m}) = \begin{bmatrix}
			-2 m_y \\
			0 \\
			-\gamma 
		\end{bmatrix}, \quad D_P(\mathbf{m}) = \begin{bmatrix}
		-2 & 0 & 0\\
		0 & 0& 0\\
		0 & 0 & 0
	\end{bmatrix} \]
	In the paramagnetic phase $\gamma > 2$, $ P $ has only one minimum on $ B_1 $ at $\mathbf{m}_0 = (0, 0, 1)^T$. The eigenvalues of the orthogonal part $  D_P^{\perp}(\mathbf{m}_0) $ of the Hessian are  $\omega_1 = -2 + \gamma $, $ \omega_2 =  \gamma $ and $|\nabla P( \mathbf{m}_0 )| = \gamma $.  Theorem~\ref{thm:pol} then yields: 
		$$ \gap H = 2 \sqrt{\gamma(\gamma-2)} + o(1) . $$  
The gap closes like a square root  close to the critical point $\gamma =2$. This is the end-point (at $ \beta = \infty $) of the critical line which separates the ferromagnetic phase of the quantum Curie-Weiss model at low temperatures from the paramagnetic phase at high temperature or large transversal field (cf.~\cite{Bapst_2012}).

		In the ferromagnetic phase $|\gamma| < 2$, minima of $ P $ are found at $ \mathbf{m}_0^\pm =( \pm \sqrt{1-\gamma^2/4}, 0 , \gamma/2)^T$. The eigenvalues of the orthogonal part $  D_P^{\perp}(\mathbf{m}_0^\pm) $ of the Hessian are  $\omega_1 =  2 ,\omega_2 = 2 (1-\gamma^2/4) $ and $|\nabla P( \mathbf{m}_0^\pm )| =2$. If we ignore the degeneracy of these two minima for the moment and pretend that only the positive solution $  \mathbf{m}_0^+ $ exists, this leads to the expression 
$ 4 \sqrt{1-\gamma^2/4} $ in the right side of~\eqref{eq:gappol}.
		As will be explained below Theorem~\ref{thm:mmin}, the gap vanishes in this phase due to the degeneracy of the two minima. What is calculated here is in fact the gap between the second excited state and the ground state.

\section{Low energy spectra for operators with regular symbols}

Our result on the gap  and the  low-energy spectrum applies to a more general class of mean-field quantum spin Hamiltonians 
than just non-commutative self-adjoint polynomials of the total spin. Next, we describe this class, which involves operators defined via their upper symbols. Besides the quest for generality, the semiclassical methods used in the proof of our results for Hamiltonians as in~\eqref{eq:Ham}  already naturally  brings up such operators with more general symbols.

\subsection{Assumptions}\label{sec:ass}
\begin{description}
\item[A0]  We assume that 
$ H $ is block diagonal with respect to the orthogonal decomposition \eqref{def:blockH} of $ \mathcal{H}_N $,
\begin{equation}\label{eq:Hblock}
H =  \bigoplus_{J=\frac{N}{2} - \lfloor \frac{N}{2}\rfloor}^{N/2} \bigoplus_{\alpha=1}^{M_J}  H_{J,\alpha} 
\end{equation}
 with self-adjoint blocks $  H_{J,\alpha}  $ acting on a copy of $ \mathbbm{C}^{2J+1} $. 
 Moreover, there is a twice continuously differentiable symbol $ h: B_1 \to \mathbb{R} $ such that 
all blocks are uniformly approximable in operator norm on $ \mathbb{C}^{2J+1} $ to order one as $ N \to \infty $:
\begin{equation}\label{ass:symbol}
	\max_{J,\alpha}  \left\| H_{J,\alpha}  -    \frac{2J+1}{4\pi} \int  N h\Big(\frac{2J}{N}  \mathbf{e}(\Omega) \Big) \,  \big| \Omega, J \rangle \langle \Omega, J \big| \, d\Omega  \right\| \leq \mathcal{O}(1) ,
\end{equation}
where the maximum runs over $ \alpha  \in \{ 1, \dots , M_{N,J} \}  $ and $ J \in \{ \frac{N}{2} - \lfloor \frac{N}{2}\rfloor, \dots , N/2\} $.
\end{description}
For our semiclassical analysis, we introduce subspaces associated with a fixed block $(J,\alpha) $ 
and any direction defined by  $  \mathbf{0} \neq \mathbf{m}_0 \in B_1 $:
$$  
\mathcal{H}^K_{J}( \mathbf{m}_0) = \spa \left\{ \psi \in \mathbb{C}^{2J+1} \ | \  \mathbf{m}_0 \cdot \mathbf{S} \ \psi = | \mathbf{m}_0| \ (J-k) \psi \; \mbox{for some $ k \in \{0, , 1 , \dots , K \} $} \right\} . 
$$
The associated orthogonal projection on $ \mathbb{C}^{2J+1}$ will be denoted by $ P_{J}^{K}(\mathbf{m}_0) $.  
We will work with quadratic approximations at $  \mathbf{m}_0  $ defined in terms of the matrix-valued self-adjoint second-order Taylor polynomial associated with $ h $ and  the total spin $   \mathbf{S}  $ on $ \mathbb{C}^{2J+1} $:
	\begin{equation}\label{def:quadratic}
	Q( \mathbf{m}_0) :=  N h( \mathbf{m}_0) \mathbbm{1} + \left(2 \mathbf{S} - N  \mathbf{m}_0 \right) \cdot  \nabla h( \mathbf{m}_0) + \frac{2}{N} (\mathbf{S} - N \mathbf{m}_0/2 ) \cdot D_h( \mathbf{m}_0)  (\mathbf{S} - N \mathbf{m}_0/2 ) .
  \end{equation}
 The operators $(  \mathbf{S} - N \mathbf{m}_0/2)/\sqrt{ J_N(\mathbf{m}_0 )} $ with $ J_N(\mathbf{m}_0 ) \coloneqq  |\mathbf{m}_0 | N/2 $  are the fluctuation operators (cf.~\cite{MN04,MN05}) with respect to the coherent state $ | \Omega_0 , J_N(\mathbf{m}_0 ) \rangle $, where $ \Omega_0 $ is the spherical angle of $\mathbf{e}_{\mathbf{m}_0}$. 
We assume the approximability of $ H_{J,\alpha} $ with $ J $ close to $  J_N(\mathbf{m}_0 )  $ either solely for the set of minima  $ \mathcal{M} \subset B_1 $ or globally at every point. 
\begin{description}
\item[A1] 
We assume that there is a continuously differentiable $ \kappa: B_1 \to \mathbb{R} $ and diverging 
sequences $ K_N $, $ \overline{K}_N \in \mathbb{N} $ such that for all  minima $ \mathbf{m}_0 \in  \mathcal{M}  $:
\begin{equation}\label{ass:quadr}
		\max_{ | J - J_N(\mathbf{m}_0 )| \leq \overline{K}_N} \max_{\alpha  } \ \left\|  \left[ \kappa\left( \mathbf{m}_0 \right) \mathbbm{1}  + Q\left(\mathbf{m}_0\right) - H_{J,\alpha} \right]  P_{J}^{K_N}( \mathbf{m}_0 ) \right\| = o(1) .
	\end{equation}
where Landau's $ o(1) $ stands for a null sequence as $ N \to \infty $. 
\item[A1']  We assume that there is a continuously differentiable $ \kappa: B_1 \to \mathbb{R} $ and a
diverging sequence $ K_N\in \mathbb{N} $ such that:
\begin{equation}\label{ass:quadr2}
		\max_{ \mathbf{m}_0 \in B_1} \max_{ | J - J_N(\mathbf{m}_0 )| \leq 1} \max_{\alpha  } \ \left\|  \left[ \kappa\left( \mathbf{m}_0 \right) \mathbbm{1}  + Q\left(\mathbf{m}_0\right) - H_{J,\alpha} \right]  P_{J}^{K_N}( \mathbf{m}_0 ) \right\| = o(1) .
	\end{equation}
\end{description}

\bigskip
 
 \noindent
Before moving on to the main results, let us add a few comments. 
 \begin{enumerate} 
   \item The reason for including an off-set function $ \kappa $ in the quadratic approximations~\eqref{ass:quadr} and~\eqref{ass:quadr2} is that the symbol $ h $ is only assumed to approximate $ H $ up to order one, cf.~\eqref{ass:symbol}. However, our main results address the spectrum exactly to this order. 
 \item 
 In case the minimum $  \mathbf{m}_0 \in \mathcal{M} $ is in the interior of the ball, $ | \mathbf{m}_0 | < 1 $, then $  \nabla h( \mathbf{m}_0)  =   \mathbf{0} $. The second term in the quadratic approximation~\eqref{def:quadratic} is hence absent. In case $ | \mathbf{m}_0 | = 1 $, the gradient either vanishes or points to the center of the ball, $ \nabla h( \mathbf{m}_0)  = - | \nabla h( \mathbf{m}_0) | \ \mathbf{e}_{\mathbf{m}_0} $. The second term in~\eqref{def:quadratic}  then equals $2   | \nabla h( \mathbf{m}_0) | ( \mathbf{e}_{\mathbf{m}_0} \cdot   \mathbf{S}  - N/2) $. 
 \item If \textbf{A1} or \textbf{A1'} hold for diverging sequences $ K_N, \overline{K}_N $, then they also hold for any such sequences, which are upper bounded by $ K_N, \overline{K}_N $. 
\end{enumerate}
The projections corresponding to the subspaces $ \mathcal{H}^K_{J}( \mathbf{m}_0) $ are chosen such that 
\begin{align}\label{eq:norm1}
\left\| \left(\mathbf{e}_{\mathbf{m}_0} \cdot \mathbf{S} - J_N( \mathbf{m}_0)\right) P_{J}^{K}(\mathbf{m}_0) \right\|  &= \max_{k \in \{0, \dots, K\}} \left| J - k - J_N(\mathbf{m}_0) \right| \leq  K +  \left| J - J_N(\mathbf{m}_0) \right| .
\end{align}
To estimate the norm of the spin operator $ Q_\perp \mathbf{S} $ projected to orthogonal directions, $  Q_\perp = \mathbbm{1}_{\mathbb{R}^3} -  \mathbf{e}_{\mathbf{m}_0}^T \mathbf{e}_{\mathbf{m}_0} $ perpendicular to $ \mathbf{m}_0 $, it is convenient to introduce a coordinate system. When focusing on a patch around one point, we may always assume without loss of generality that $ \mathbf{m}_0 = (0,0,|\mathbf{m}_0|)^T $. This can be accomplished by the unitary rotation in~\eqref{def:cs}, for which $ U(\Omega_0)^* \ \mathbf{e}_{\mathbf{m}_0} \cdot \mathbf{S} \ U(\Omega_0) = S_z $ if $ \Omega_0$ denotes the spherical coordinates of $  \mathbf{e}_{\mathbf{m}_0} $. In this coordinate system, the spin operators in the two orthogonal directions are given by $ S_x $ and $ S_y $, and the range of the projections $ P_{J,\alpha}^{K}(\mathbf{m}_0)  $ are spanned by the canonical orthonormal eigenbasis of $ S_z $ on $ \mathbb{C}^{2J+1} $, i.e. $ S_z  | J- k \rangle = (J- k) | J- k \rangle $ for all $ k \in \{0, 1, \dots, 2J\} $.  We recall that both $ S_x $ and $ S_y $ are tridiagonal matrices in terms of this basis:
\begin{align}\label{eq:MatrixSorth}
\langle J - k |  S_x | J - {k^\prime}\rangle & =  i^{k^{\prime}-k} \langle J - k|  S_y | J - {k^\prime}\rangle  = \sqrt{\frac{2J \max\{ k, {k^\prime}\} - k {k^\prime} }{4}} \delta_{|k'-k|,1 } . 
\end{align}
Therefore, for any $ d \in \mathbb{N} $ there is some $ C_d < \infty $ such that for both $ \xi \in \{ x,y \} $:
\begin{equation}\label{eq:dthpowerS}
\max_{J\leq N/2} \max_{\min\{k,k'\} \in \{ 0, \dots , K\} } \left| \langle J - k |  S_\xi^d | J - {k^\prime}\rangle \right| \leq C_d  \left(N K \right)^{\frac{d}{2}}   1[ |k'-k|\leq d ] .
\end{equation} 
This renders evident that on $ \mathcal{H}^K_{J}( \mathbf{m}_0)$ the scaled operators  $ N ( 2 S_\xi/N)^d $  are negligible if $ d \geq 3 $ and  $ K = o(N^{1/3} ) $. 

The above arguments also show that our assumptions are tailored to apply to~\eqref{eq:Ham} with an arbitrary non-commuting self-adjoint polynomial $ P $. 
\begin{theorem}\label{thm:polass}
For $ H =  N \textrm{P}\big(\tfrac{2}{N}   \mathbf{S} \big) $ on $ \mathcal{H}_N $ with any non-commuting self-adjoint polynomial $ P $, all assumptions~\textrm{\textbf{A0, A1, A1'}} are satisfied with $ h = P $ and $ \overline{K}_N  = o(N^{2/3})$,  $ K_N = o(N^{1/3}) $. 
\end{theorem}
\begin{proof}

Thanks to~\eqref{eq:uppersymbclasss}, an approximate symbol of the Hamiltonian is indeed the polynomial, $ h = P $. 

Let $ \mathbf{m}_0 \in B_1 $ be arbitrary, not necessarily a minimum of $ P $. Without loss of generality we may assume $ \mathbf{m}_0 = (0,0,|\mathbf{m}_0|)^T $ and that the Hamiltonian is of the form $ H = N \textrm{P}\left( s_N(1) , s_N(2) ,  s_N(3)\right) $ with
$$
s_N(1) := \frac{2}{N}  S_x , \quad s_N(2):= \frac{2}{N} S_y, \quad s_N(3) := \frac{2}{N} (S_z - J_N(\mathbf{m}_0) ). 
$$
The polynomial $  \textrm{P} $  is a linear combination of monomial products of a fixed order $ d \in \mathbb{N} $. Due to non-commutativity of matrix multiplication such monomials include all products of the form $$ \Pi_N(\mathbf{i}) := s_N(i_1) s_N(i_2) \dots s_N(i_d) $$ with an arbitrary choice of ordered indices $\mathbf{i} =( i_1, \dots , i_d) \in \{ 1,2,3\}^d $. 
Up to order $ d \leq 2 $ and due to the assumed self-adjointness of the non-commutative product, they coincide with the terms in the definition~\eqref{def:quadratic} of $ Q( \mathbf{m}_0) $. 

Each of the above matrices $  \Pi_N(\mathbf{i}) $ is block diagonal on $ \mathcal{H}_N $ and, when restricted to a copy of $ \mathbb{C}^{2J+1} $,  at most $2d+1$-diagonal. We may therefore estimate similarly as in~\eqref{eq:dthpowerS}   for all  $ J \leq N/2 $ and any choice of indices $ \mathbf{i} \in \{ 1,2,3\}^d$:
\begin{align}\label{eq:productkernel}
& \max_{\min\{k,k'\} \in \{ 0, \dots , K_N\} } \left| \langle J - k | \ N   \Pi_N(\mathbf{i}) \ | J - {k^\prime}\rangle \right| \notag  \\ 
& \quad \leq C_d \ N^{1-d}  \max_{m\in\{ 0, \dots, d\} } \left\{ \sqrt{N (K_N+d)}^{d-m} \left( | J- J_N(\mathbf{m}_0) | + K_N+d \right)^{m} \right\}  1[ |k'-k|\leq d ] \notag \\
& \quad   \leq  C_d \left[ \frac{\sqrt{K_N}^d}{\sqrt{N}^{d -2} }+ \frac{|J-J_N(\mathbf{m}_0) |^d + K_N^d }{N^{d-1} }  \right] 1[ |k'-k|\leq d ] ,
\end{align}
with a constant $  C_d < \infty $, which changes from line to line, and which is independent of $ J $. 
For any $ d \geq 3 $, the right side vanishes if $ K_N = o(N^{1/3}) $ and $ |J-J_N(\mathbf{m}_0) | = o(N^{2/3})$. Since
\begin{equation}\label{eq:normfrom kernel}
\max_\alpha \left\|   \Pi_N(\mathbf{i}) P_{J}^{K_N}( \mathbf{m}_0 ) \right\| \leq (2d+1) \mkern-10mu  \max_{\min\{k,k'\} \in \{ 0, \dots , K_N\} } \mkern-10mu  \left| \langle J - k |   \Pi_N(\mathbf{i})  | J - {k^\prime}\rangle \right|   ,
\end{equation}
any monomial of order $ d \geq 3 $  indeed  does not contribute in~\eqref{ass:quadr}. 
\end{proof}
Using similar estimates and under some more restrictive assumptions on $ K_N $ and $ \overline{K}_N $, one may replace $ Q( \mathbf{m}_0)  $ in assumption~\eqref{ass:quadr}  by the  second-order polynomial
\begin{equation}\label{eq:defQhat}
\widehat Q( \mathbf{m}_0) \coloneqq  N h( \mathbf{m}_0) \mathbbm{1} + \left(2 \mathbf{S} - N  \mathbf{m}_0 \right) \cdot  \nabla h( \mathbf{m}_0) + \frac2{ N}   (Q_\perp \mathbf{S}) \cdot D_h( \mathbf{m}_0) (Q_\perp\mathbf{S} )  ,
\end{equation}
which only involves the projection $ Q_\perp $ of the Hessian. This means that the fluctuation operator in the radial direction of $  \mathbf{m}_0 $ is negligible. 
\begin{lemma}\label{lem:altQ}
If $ K_N = o(N^{1/3}) $ and $  \overline{K}_N^2 K_N = o(N) $, then 
\begin{equation}
\sup_{ | J - J_N(\mathbf{m}_0 )| \leq \overline{K}_N } \sup_{\alpha  }  \left\| \left[Q( \mathbf{m}_0) - \widehat Q( \mathbf{m}_0) \right] P_{J}^{K_N}(\mathbf{m}_0)  \right\| = o(1) .
\end{equation}
\end{lemma}
\begin{proof}
Using the same coordinate system and notation as in the proof of Theorem~\ref{thm:polass}, 
the difference $ Q( \mathbf{m}_0) - \widehat Q( \mathbf{m}_0) $ is a linear combination of the five monomials of the form $ N  \Pi_N(\mathbf{i})  $ with 
$ \mathbf{i} =(i_1,i_2) \in \{ (1,3), (3,1), (2,3), (3,2), (3,3) \} $. 

By~\eqref{eq:norm1}, we have $  \left\| N  \Pi_N(3,3) P_{J}^{K_N}(\mathbf{m}_0)  \right\| \leq 8 \ (  | J - J_N(\mathbf{m}_0 )|^2 + K_N^2 ) / N = o(1) $. In all other cases, we estimate similarly as in~\eqref{eq:productkernel}. For example:
$$
\max_{k^\prime \in \{0, \dots, K_N\}}   \left| \langle J - k |  N   \Pi_N(3,1) | J - {k^\prime}\rangle \right|  \leq    \sqrt{\frac{K_N+1}{N}} \left( | J- J_N(\mathbf{m}_0) | + K_N+1 \right)   1[ |k'-k|\leq 1 ] , 
$$
which by~\eqref{eq:normfrom kernel} leads to $   \left\| N  \Pi_N(3,1) P_{J,\alpha}^{K_N}(\mathbf{m}_0)  \right\| \leq o(1) $. The remaining terms are estimated similarly. 
\end{proof}

\subsection{The case of a unique minimum} 
The following is our main result for mean-field Hamiltonians whose approximate symbol has a unique minimum, which is at the surface of the unit ball.
\begin{theorem}\label{thm:unimin}
Assuming  \textbf{A0} and \textbf{A1} and that the symbol $ h $ has a unique global minimum at  $  \mathbf{m}_0 \in S^2 $ where   $ \nabla h(\mathbf{m}_0) \neq \mathbf{0} $ and $ \det   D_h^{\perp}(\mathbf{m}_0) > 0 $:
\begin{enumerate}
\item
 the ground state $\psi_0$ of $ H $ on $ \mathcal{H}_N $ is unique (up to phase) and contained in the subspace with maximal total spin $J= N/2$. In terms of the eigenstates $| N/2 - m \rangle $ of $ \mathbf{m}_0 \cdot \mathbf{S} $ in that subspace, we have for any $ m \in \mathbb{N}_0 $:
 \begin{equation}\label{eq:eigenvect}
  \langle J-m | \psi_0 \rangle = \omega^{1/4} \sqrt{\frac{2}{(\omega +1) n!}} \left(\sqrt{\frac{\omega-1}{2(\omega +1)}}\right)^n H_n(0) + o(1) ,
 \end{equation}
where $H_n$ denotes the $n$-th Hermite polynomial, and the ground state energy is given by 
        \begin{equation}\label{eq:gsenergy}
        	E_0(H) = N h( \mathbf{m}_0) + \kappa( \mathbf{m}_0) - | \nabla h(\mathbf{m}_0) |  + \sqrt{\det   D_h^{\perp}(\mathbf{m}_0)} +o(1).
        \end{equation}
        \item for any energy below $ E_0(H) + \Delta $ with $ \Delta  > 0 $ fixed but arbitrary, the eigenvalues of $ H $ stem from blocks $ (J,\alpha) $ with $ k = N/2 - J \in \mathbb{N}_0 $ fixed. When counted with multiplicity, these low-energy eigenvalues of $ H_{N/2-k,\alpha} $ for $ k \in \mathbb{N}_0 $ asymptotically coincide up to $ o(1) $ with the points in
        \begin{equation} \label{eq:allev}
         N h( \mathbf{m}_0) + \kappa( \mathbf{m}_0) + (2k-1)\  | \nabla h(\mathbf{m}_0) | + (2m+1) \sqrt{\det   D_h^{\perp}(\mathbf{m}_0)} , 
        \end{equation}
        with $ m \in \mathbb{N}_0 $.  
    The spectral gap of $ H $ is 
    \begin{equation}\label{eq:gapH}
    	\gap H =   2 \min \left\{ |\nabla h( \mathbf{m}_0 ) |, \sqrt{\det   D_h^{\perp}(\mathbf{m}_0)} \right\}+o(1).
    \end{equation}
    \end{enumerate}
\end{theorem} 
The proof of Theorem~\ref{thm:unimin} is the topic of Section~\ref{sec:pfunim}.  As can be inferred from there, the error estimates hiding in the $ o(1) $-terms are made up from two sources: first, the accuracy  of the assumed quadratic approximation~\eqref{ass:quadr}, and second, the cumulative subsequent error estimates. Moreover, an asymptotic form of all eigenfunctions of $ H_{J,\alpha} $ on the outer blocks and not only the ground state~\eqref{eq:eigenvect} is found there. 
Although the ground state is simple and the same applies to the oscillator spectrum~\eqref{eq:allev} of $ H_{J,\alpha} $ for fixed $ (J,\alpha) $, the multiplicity $ M_{N,J} $ in~\eqref{eq:Hblock} will cause the eigenvalues of $ H $ to occur approximately to order $ o(1) $ with these multiplicities.  \\

Before moving on to the general case, let us put Theorem~\ref{thm:unimin}  in the context of available results. 
\begin{enumerate}
\item 
The spectra of quadratic mean-field Hamiltonians such as the Lipkin-Meshkov-Glick  model can be described exactly  through  Bethe-Ansatz-type equations 
\cite{Turbiner:1988aa,Pan:1999aa,Ortiz:2005aa}. Since the latter are in the most general case hard to solve, much attention has been given to   finding approximate semiclassical solutions \cite{Ribeiro:2008aa,Bulekov:2021aa}. In view of this, it is worth emphasizing that the above theorem is (through Theorem~\ref{thm:polass}) applicable to all $ H $ of polynomial form~\eqref{eq:Ham}. Their low-energy spectra are proven to agree with that of the associated quadratic term. The latter turns out to produce the harmonic oscillator spectrum~\eqref{eq:allev}, in which the frequency is determined by the Hessian in the spherical geometry.   
The spectrum of such general mean-field Hamiltonians has so far only been determined  to a coarser order $ N $ in~\cite{Ven:2020aa} (cf.~\eqref{eq:VV20}) and not on the fine scale $ o(1) $.  
\item
Expressions for the spectral gap of certain polynomial mean-field quantum-spin Hamiltonians have been derived on the level of rigor of theoretical physics in~\cite{Ribeiro:2008aa,Bapst_2012}. These works assume that the low-energy spectrum of the relevant blocks $ H_{J,\alpha} $ are equally spaced, and argue that the gap is proportional to its inverse density of states at the ground state. 
A-posteriori and thanks to the proven equal spacing of the low-energy eigenvalues of  $ H_{N/2 - k ,\alpha}  $, this is correct in case the minimum in~\eqref{eq:gapH} is attained at $  \big(\det   D_h^{\perp}(\mathbf{m}_0)\big)^{1/2}$. It hence works in case of the quantum Curie-Weiss model in the paramagnetic phase. However, in case the minimum in~\eqref{eq:gapH}  is found at $   |\nabla h( \mathbf{m}_0 )|  $ and hence stems from the second-most outer shell, this strategy fails. 
\end{enumerate}

Theorem~\ref{thm:unimin} should be contrasted with the case that the minimum of the symbol is found strictly inside the unit ball. In this case, the spectral gap vanishes, since the following theorem shows that all blocks $ (J,\alpha) $ with $ | J - J_N(\mathbf{m}_0) | \leq o(\sqrt{N }) $ have the same ground-state energy up to an $ o(1) $-error. 
\begin{theorem}\label{thm:unimin2}
Assuming  \textbf{A0} and \textbf{A1'}and that the symbol $ h \in C^3 $ has a unique global minimum at  $  \mathbf{m}_0 \in B_1 $
with $0 < | \mathbf{m}_0| < 1$ and $ D_h( \mathbf{m}_0) > 0 $. Then:
\begin{enumerate}
\item  any ground-state eigenvector of $ H $ on $ \mathcal{H}_N $ is contained in a subspace with total spin $ J $ with $ | J - J_N(\mathbf{m}_0) | \leq \mathcal{O}(\sqrt{N })$ and 
\begin{equation}\label{eq:gsball}
E_0(H) = E_0(H_{J,\alpha}) =  N h(\mathbf{m}_0)  + \kappa( \mathbf{m}_0) +  | \mathbf{m}_0|  \sqrt{\det   D_h^{\perp}(\mathbf{m}_0)} + o(1) .
\end{equation}
\item for any $ J $ with $ | J - J_N(\mathbf{m}_0) | \leq o(\sqrt{N }) $ the ground-state energy $ E_0(H_{J,\alpha})  $ is still given by~\eqref{eq:gsball}. 
\end{enumerate}
\end{theorem}
The proof largely builds on the techniques of the proof of Theorem~\ref{thm:unimin} and is found in Section~\ref{sec:pfunim2}. The techniques allow in fact to determine the whole low-energy spectrum of every block $ H_{J,\alpha} $ with $ | J - J_N(\mathbf{m}_0) | \leq o(\sqrt{N }) $. 

\subsection{The case of a finite number of minima} 
Our last main result concerns the case of symbols with finitely many minima on the surface of the quantum sphere. 
\begin{theorem}\label{thm:mmin} 
Assuming \textbf{A0} and \textbf{A1} and that the symbol $ h $ has $ L $ global minima at $ \{ \mathbf{m}_1, \dots , \mathbf{m}_L \} \subset S^2 $ where at each minimum $ \nabla h(\mathbf{m}_l) \neq \mathbf{0} $ and $ \det   D_h^{\perp}(\mathbf{m}_l) > 0 $:
\begin{enumerate}
\item 
 any ground-state eigenvector of $ H $ on $ \mathcal{H}_N $ is  contained in the subspace with maximal total spin $J= N/2$. The ground-state energy is
 $$
 E_0(H) = \min_{l \in \{1,\dots,L\}} \left[  h( \mathbf{m}_l) + \kappa( \mathbf{m}_l) - | \nabla h(\mathbf{m}_l) |  + \sqrt{\det   D_h^{\perp}(\mathbf{m}_l)} \right] +o(1).
 $$
\item for any energy below $ E_0(H) + \Delta $ with $ \Delta  > 0 $ fixed but arbitrary, the eigenvalues of $ H $ stem from blocks $ (J,\alpha) $ with $ k = N/2 - J \in \mathbb{N}_0 $ fixed but arbitrary. When counted with multiplicity, these low-energy eigenvalues of $ H_{N/2-k,\alpha} $ for $ k \in \mathbb{N}_0 $ asymptotically coincide up to $ o(1) $ with the points in
        \begin{equation} \label{eq:allev2}
         N h( \mathbf{m}_l) + \kappa( \mathbf{m}_l) + (2k-1)\  | \nabla h(\mathbf{m}_l) | + (2m+1) \sqrt{\det   D_h^{\perp}(\mathbf{m}_l)} , 
        \end{equation}
        with $ m \in \mathbb{N}_0 $ and $ l \in \{1,\dots, L \} $. 
\end{enumerate}
\end{theorem}
The proof, which in comparison to Theorem~\ref{thm:unimin} poses the additional difficulty of controlling the tunnelling between minima,  is found in Section~\ref{sec:prmmin}.\\

The theorem allows for degeneracies in the spectrum already at the level of the ground state. The quantum Curie-Weiss model $P(\mathbf{m})  =  -m_x^2-\gamma m_z$ in the ferromagnetic phase $  |\gamma| < 2 $ is an example. The gradient's norm $ |\nabla P(\mathbf{m}_0^\pm)| = 2 $ and $ \det D_P^\perp(\mathbf{m}_0^\pm) = 4- \gamma^2 $ agree for the two minima $ \mathbf{m}_0^\pm \in S^2 $. Therefore, all the low-energy eigenvalues described by~\eqref{eq:allev2} are approximately doubly degenerate. In this situation the formula~\eqref{eq:gappol} yields the gap of the nearly $ L $-fold degenerate ground state ($L=2$ for Curie-Weiss)  to the next energy levels. 
Our proof enables to show that the level splitting due to tunnelling through a macroscopic barrier from one minimum to the other is smaller than any polynomial in $ N^{-1} $. As demonstrated numerically in~\cite{VGR20}, the quantum Curie-Weiss's ferromagnetic phase  exhibits the 'flea on the elephant phenomenon'  \cite{Jona-Lasinio:1981aa}, i.e., the sensitivity of the ground-state function to perturbations. It might be interesting to combine the techniques in this paper with \cite{Helffer:1984aa,Simon:1985aa} for a proof of this. \\

Let us conclude the main part of the paper with a general outlook. 
Having derived precise low-energy asymptotics of eigenvalues in terms of quadratic approximations and using techniques as in~\cite[Sec.~5]{MW23}, our techniques should extend to derive the fluctuations of the free energy not only at $ \beta = \infty $, but also for finite temperature. 

It would also be interesting to investigate the dynamical properties of the mean-field Hamiltonians. Coherent wavepackets evolve semiclassically~\cite{Frohlich:2007aa}. The latter work also connects to the question of whether the description of the low-energy spectra in the present paper are helpful in the analysis of models which become semiclassical only in a Kac-type scaling limit (cf.~\cite{Pres09}).

\section{Semiclassical analysis}

This section is dedicated to the proofs of Theorems~\ref{thm:unimin}-\ref{thm:mmin}.  The proofs combine projection techniques  in a Schur-complement analysis, which in the present paper constitutes of the following simple principle. 
\begin{proposition}\label{prop:Schur}
Let $ A $ be a bounded self-adjoint operator on a Hilbert space $ \mathcal{H} $, and $ E $ and $ F $ orthogonal projections with $ E + F = \mathbbm{1}_\mathcal{H} $. Assume that $ a < \inf \spec F A F  $ and let $ R(a) = ( F AF -  a F)^{-1} $ stand for the block inverse on $ F\mathcal{H} $. Then 
$$
 a \in \spec A \quad \mbox{if and only if} \quad  0 \in \spec EAE - a E - EAFR(a) FAE .
$$
In particular, the eigenvalues  $ \alpha_0(A) \leq \alpha_1(A) \leq \dots $ of $ A $ (counted with multiplicities) and the respective eigenvalues of $ EAE $ satisfy
$$
\left| \alpha_j(A) - \alpha_j(EAE) \right| \leq \frac{\| E  A F \|^2}{\dist(\spec FAF, a) }
$$
provided $ \alpha_j(A) < a < \spec FAF$. 
\end{proposition}
The proof is elementary. Extensions of this statement can be found in \cite{SZ03,Zh05,Zwo12,GustSig11}. \\

Through suitably defined projections, our proof  of Theorems~\ref{thm:unimin}-\ref{thm:mmin} focuses the spectral analysis of the mean-field Hamiltonian on subspaces associated with patches around the classical minima of its symbol. More precisely, we proceed by a two-step localization procedure. To illustrate the main idea, let us focus on the proof of  Theorem~\ref{thm:unimin}. For a first coarse localization and at a fixed value of the total spin $ J $, we  restrict attention to the subspaces $ \mathcal{H}^{K_N}_{J}( \mathbf{m}_0)  $ with $ K_N $ increasing with $ N $ and $ \mathbf{m}_0 $  the unique minimum of the symbol. On these increasing nested subspaces, we construct a sequence of approximate Hamiltonians, which result from an explicit limit operator of the fluctuation operators.  In a second localization, the Schur-complement analysis  of Proposition~\ref{prop:Schur} is then applied with $ E $ the subspace spanned by an arbitrary, but finite number of eigenspaces of the approximate limit operator, which turns out to coincide with a harmonic oscillator. The off-diagonal terms featuring in~Proposition~\ref{prop:Schur} are estimated with the help of exponential decay estimates on the explicit eigenfunctions of the limit operator. 

In case of Theorem~\ref{thm:mmin}, we proceed similarly. An additional challenge is to control the tunneling  between patches of different minima. This is done here using the exponential decay of coherent-state inner products, which we derive in Appendix~\ref{app:cs1}. Let us stress that our method does not rely on sharp semiclassical tunneling estimates in terms of the Agmon metric, which are available e.g.\ for related problems of multi-well tunneling for Schr\"odinger operators~\cite{Helffer:1984aa,Simon:1985aa}. 

To the best of our knowledge, the above proof strategy is new in the present context of a spectral analysis of mean-field quantum spin systems. In a broader context. the Schur complement methods or more generally, Gushin's method, have been employed before in semiclassical analysis, e.g.\ for stability analysis and, under the name Feshbach-Schur method, in perturbation theory in quantum mechanics,  cf.~\cite{SZ03,Zwo12,GustSig11}.\\

We start with a description of the limit operators. The proof of Theorems~\ref{thm:unimin}-\ref{thm:mmin} follow in the subsequent sections.

\subsection{Limit of fluctuation operators}\label{sec:limop} 

We start our proof by introducing two operators $L_x$ and $L_y$ on the Hilbert space $\ell^2(\nn_0)$, which turn out to be unitarily equivalent to the position and momentum operator on $L^2(\rr)$. This enables us to determine the spectrum and eigenfunctions of the operator
\[ D = \omega^2 L_x^2+ L_y^2  ,  \]
which is equivalent to a harmonic oscillator with frequency $\omega$. In the following subsections, we show that the Hamiltonians described in Section~\ref{sec:ass}  indeed converges in a sense to be specified locally to an operator of the form $D$.  This is the key to determine their low energy spectra explicitly. \\

To set the stage, we define $L_x$ and $L_y$ via their matrix elements in terms of the canonical orthonormal bases in $ \ell^2(\mathbb{N}_0) $:
\begin{equation}\label{def:L}
		\langle k | L_x | {k^\prime} \rangle =  i^{k^\prime-k} 
		\langle k | L_y | {k^\prime} \rangle = \sqrt{\frac{\max\{k,k^\prime\}}{2}} \delta_{|k-k^\prime| =1}.
\end{equation}
They  give rise to essentially self-adjoint operators on the dense domain 
\[ c_{00} \coloneqq \{ (x_n)_{n \in \nn}\, | \, \exists \, N \in \nn \text{ such that } x_n = 0 \, \, \forall \, n \geq N \}.  \]
Moreover, by an elementary calculation 
\begin{equation}\label{eq:comm} [L_x,L_y] | k \rangle \coloneqq (L_x L_y - L_y L_x) | k \rangle  = i | k \rangle ,  \end{equation}
i.e., $L_x$ and $L_y$ satisfy the canonical commutation relations. The following proposition is essentially a consequence of this observation.

\begin{proposition}\label{prop:harm}
Let $\omega \geq 1$. Then, $D = \omega^2 L_x^2+  L_y^2$ is a positive, essentially self-adjoint operator on the domain $ c_{00} \subset \ell^2(\mathbb{N}_0)$ with spectrum 
	 \begin{equation}\label{eq:harmosc}
	 	\spec D = \{(2k+1) \omega \, | \, k \in \nn_0 \} .
	 \end{equation}
Every point in the spectrum is a non-degenerate eigenvalue.
The ground-state  $\psi_0$ is 
 \begin{equation}\label{eq:gs}
 	\langle n | \psi_0\rangle = \omega^{1/4} \sqrt{\frac{2}{(\omega +1) n!}} \left(\sqrt{\frac{\omega-1}{2(\omega +1)}}\right)^n H_n(0).
 \end{equation}
The $k$-th excited state is given by $\psi_k = (a^\dagger)^k \psi_0/\sqrt{k!} $ with the raising operator
$
 	a^\dagger\coloneqq \sqrt{\frac{\omega}{2}} \left(L_x - \frac{i}{\omega} L_y\right) 
$. 
 In particular, $ \langle n | \psi_k\rangle = 0 $ unless $ k-n $ is even in which case, we have the exponential decay estimate
 \begin{align}\label{eq:expdecay}
 \left|  \langle n | \psi_k\rangle  \right|^2 & \leq  \sqrt{\frac{\omega}{\pi}} \frac{2^{2k+1}}{\omega+1}  \frac{n^{k-\frac{1}{2}}}{ k!}   \left( \frac{\omega-1}{\omega+1}\right)^{n-k } ,\quad  \mbox{if $ n\geq 2k $,} && \\ 
  \left|  \langle n | \psi_k\rangle  \right|^2 & \leq   \sqrt{\frac{\omega}{\pi}}\frac{2^{2n+1}}{\omega+1}   \frac{k^{n-\frac{1}{2}}}{ n!}   \left( \frac{\omega-1}{\omega+1}\right)^{n-k } ,\quad  \mbox{if $ k\geq 2 n $.} &&
  \label{eq:expdecay2}
 \end{align}
\end{proposition}
The value $H_n(0)$ is explicit:
\begin{equation} \label{eq:Hermite0}
H_n(0) = \begin{cases} (-1)^{n/2} \frac{n!}{(n/2)!}& \text{ if } $n$ \text{ even } \\
                    0 &  \text{ if } $n$ \text{ odd. }\end{cases} 
 \end{equation} 

In case $ \omega = 1 $, the eigenbasis turns out to agree with the canonical orthonormal basis, $ | \psi_k \rangle = | k \rangle $  for all $ k \in \mathbb{N}_0 $. The exponential decay estimates~\eqref{eq:expdecay} and~\eqref{eq:expdecay2} will play a crucial role in subsequent approximation results. 
\begin{proof}[Proof of Proposition~\ref{prop:harm}] 
The commutation relation \eqref{eq:comm} implies that there is a unitary $U \colon \ell^2(\nn_0) \to L^2(\rr)$ such that 
$ U L_x U^{*} = \hat{x} $ and $ U L_y U^{*} = \hat{p} $, 
where $\hat{x}$ and $\hat{p}$ denote the position and momentum operator. In particular, $U D U^{*}$ is the standard harmonic oscillator with frequency $\omega$, which yields the basic assertion on  $\spec  D $. The eigenfunctions of  $ U D U^{*} $ are known to be given by 
\[ \varphi_n^\omega(x) = \frac{1}{\sqrt{2^n n!}} \left( \frac{\omega}{\pi} \right)^{1/4} e^{-\omega x^2/2} H_n(\sqrt{\omega} x).   \]
This unitary equivalence also proves the ladder operator representation for the excited states (cf.~\cite{Hall13}). 
Since $U| n \rangle  = | \varphi_n^{1} \rangle $, we seek for a representation of  $\varphi_m^\omega$ and, in particular $ m = 0 $, in terms of the $\omega = 1$ eigenfunctions:
$$
\langle \varphi_n^{1}  | \varphi_m^{\omega} \rangle  = \frac{\omega^{1/4}}{\sqrt{2^{n+m} \pi  n! m! } }   \int  e^{-(\omega+1) x^2/2 } H_n\left(x \right) H_m\left(\sqrt{\omega} x \right)   dx 
$$
Using a change of variables and subsequently the multiplication theorem for Hermite polynomials, 
$$
H_n(\alpha x) = \sum_{l=0}^{\lfloor\frac{n}{2}\rfloor} \alpha^{n-2l} (\alpha^2 -1)^l \frac{n!}{(n-2l)! l!} H_{n-2l}(x) , \quad \alpha \in \mathbb{R} ,
$$
we arrive after some elementary algebra at 
 \begin{align}\label{eq:monster}
 	\langle n | \psi_k\rangle =  \omega^{1/4} \sqrt{\frac{2 \ n! k!}{(\omega +1)}} \  \sum_{l=0}^{\lfloor\frac{n}{2}\rfloor} (-1)^{l} \frac{1\left[l + \tfrac{k-n}{2} \in \big\{0, 1 , \dots , \lfloor\frac{k}{2}\rfloor  \big\} \right]}{ (n-2l)! \ l! \ (l+ \tfrac{k-n}{2} )!} \left( \frac{2\sqrt{\omega}}{\omega+1}\right)^{n-2l} \left( \frac{\omega-1}{2(\omega+1)}\right)^{2l +\tfrac{k-n}{2} }   . 
 \end{align}
The above sum is only non-zero if $ n - k $ is even. In particular, due to the indicator function requiring $ l = n/2 $,  for $ k = 0 $  only one term in the sum survives yielding~\eqref{eq:gs} with $ H_n(0) $ replaced by~\eqref{eq:Hermite0}. 

For a proof of the exponential decay estimates~\eqref{eq:expdecay}, we start from~\eqref{eq:monster} with the triangle inequality. Since $  \frac{2\sqrt{\omega}}{\omega+1} \leq 1 $, this term can be upper bounded by one. Furthermore, since also $ \frac{\omega-1}{2(\omega+1)} \leq 1 $, the we may lower bound $ 2l +\tfrac{k-n}{2} \geq \frac{|k-n|}{2}  $, since $ l\geq 0 $ in case $ k\geq n $ and $   l\geq (n-k)/2 $ in case  $ n\geq k $. 
Therefore, in case $ k \geq n $ it remains to estimate
$$
 \sum_{l=0}^{\lfloor\frac{n}{2}\rfloor} \frac{\sqrt{n! k!}}{(n-2l)!} \frac{1}{l! (l+ \tfrac{k-n}{2} )!} \leq \sum_{l=0}^{n} \binom{n}{l} \frac{\sqrt{k!}}{\sqrt{n!} \ (\tfrac{k-n}{2} )!} = \frac{2^n}{\sqrt{n!} }  \frac{\sqrt{k!}}{(\tfrac{k-n}{2} )!} .
$$
The last ratio is then estimated with the help of standard Stirling bounds,
$$
\sqrt{2\pi m} \left(\frac{m}{e}\right)^m \exp\left(\frac{1}{12m+1}\right)\leq m! \leq \sqrt{2\pi m} \left(\frac{m}{e}\right)^m \exp\left(\frac{1}{12m}\right), \quad m \in \mathbb{N} , 
$$
together with the elementary bound $ \exp\left( - (k-n) \ln\left( 1 -  \frac{n}{k}\right) \right)  \leq 2^n $ valid for all $ k \geq 2n $. This yields~\eqref{eq:expdecay}  after some algebra. 
In case $ n \geq k $ we proceed similarly. The only difference is that the sum in~\eqref{eq:monster} starts from~$\frac{|k-n|}{2}$. 
\end{proof}

The connection of $ D $ to our models will be through its approximations defined on $ \mathbb{C}^{2J+1} $:
\begin{equation}\label{def:A_N}
D_N :=  \omega^2  \big(L^{(N)}_x\big)^2 +\big(L^{(N)}_y\big)^2    \qquad\mbox{with} \quad L^{(N)}_\xi \coloneqq \frac{S_\xi}{\sqrt{J_N}} , \quad J_N := \frac{N}{2} |\mathbf{m}_0| .
\end{equation}
In this subsection $ |\mathbf{m}_0| \in (0,1] $ is treated as a scaling parameter. 
When restricting the commutation relation, $ \big[  L^{(N)}_x ,  L^{(N)}_y \big] = i S_z / J_N $, to the subspace 
$$ \mathcal{H}^{K_N}_{J} :=\spa\left\{ | J-k \rangle \ \big| \ k \in \{ 0, \dots , K_N\} \right\}  
$$ 
with $ |J-J_N| \leq \overline{K}_N $, the commutator 
asymptotically agrees as $ N \to \infty $ with the canonical one as long as both $ \overline{K}_N $ and $ K_N $ grow only suitably slow with $ N $. 
Hence $ D_N   $ when restricted to $  \mathcal{H}^{K_N}_{J} $  is expected to be equal to the harmonic oscillator $ D $ in the limit $ N \to \infty $.  
This observation is at the heart of many works on quantum spin systems in the large $ J $ limit \cite{HeppLieb73,MN04,MN05}. To turn it into a mathematical argument in the present context, we note that through the identification $ | J-k \rangle \equiv | k \rangle $ of their canonical orthonormal basis, the Hilbert spaces $   \mathcal{H}^{K_N}_{J} \subset \mathbb{C}^{2J+1}  $ are all canonically embedded into $ \ell^2(\mathbb{N}_0 ) $. This embedding will be denoted by 
$$ I_J^{(N)}: \mathcal{H}^{K_N}_{J}\to \ell^2(\mathbb{N}_0) , \quad  \mbox{and}\quad  \overline I_J^{(N)}  : \ell^2(\mathbb{N}_0) \to  \mathcal{H}^{K_N}_{J}\ $$
stands for its corresponding projection. Thus  
$$
\overline D_{J,N} \coloneqq   \overline  I_J^{(N)} D  I_J^{(N)} 
$$
when restricted to $  \mathcal{H}^{K_N}_{J}  $ is unitarily equivalent to $\overline D_{N} \coloneqq  P_{K_N} D P_{K_N} $  when restricted to 
$$ I_J^{(N)}  \mathcal{H}^{K_N}_{J} = \spa\left\{ | k \rangle \in \ell^2(\mathbb{N}_0) \ | \ k \in \{ 0, \dots , K_N\} \right\}  \subset c_{00} $$ 
with $ P_{K_N} $ denoting the corresponding orthogonal projection in $ \ell^2(\mathbb{N}_0) $.

\begin{lemma}\label{lem:src}
\begin{enumerate}
\item
For any choice of sequences $ K_N = o(N^{1/2}) $ and $ K_N  \overline{K}_N^3 = o(N^2) $:
\begin{align}\label{eq:normDDN}
& \max_{|J-J_N|\leq \overline{K}_N} \left\| \left(\overline D_{J,N} -   D_N   \right) P_J^{K_N^-}  \right\| =  o(1) ,
\end{align}
where $ K_N^- \coloneqq K_N -2 $ and $ P_J^{K_N^-}  $ denotes the orthogonal projection onto $  \mathcal{H}^{K_N^-}_{J} \subset  \mathcal{H}^{K_N}_{J}$.  
\item 
For any sequence $ K_N \to \infty $, the operators $ \overline D_N = P_{K_N} D P_{K_N}   $ converge as $ N \to \infty $ in strong-resolvent sense to $ D $.  
Moreover, we also have the pointwise convergence of eigenvalues:
\begin{equation}\label{eq:convev}
 E_k\left(   \overline D_N \right) = E_k\left( D \right) + o(K_{N}^{-\infty}) 
\end{equation}
 for any $ k \in \mathbb{N}_0 $. If $ E_K $ and $ E_K^{(N)} $ denote the spectral projection onto the $ K \in \mathbb{N} $ lowest  eigenvalues of  $ D $  and $  \overline D_N $ on $ P_{K_N} \ell^2(\mathbb{N}_0) $, then 
\begin{equation}\label{eq:convspecprN}  \left\|  \left( E_K^{(N)} - E_K\right) E_K^{(N)} \right\| = o(K_{N}^{-\infty}) 
\end{equation}
where $ o(K_{N}^{-\infty})  $ denotes a sequence, which when multiplied by an arbitrary but fixed power of $ K_N $ goes to zero. \\
\end{enumerate}
\end{lemma}
\begin{proof}
1.~For any $ k, k^\prime \leq K_N  $, the matrix-elements of the difference in the canonical basis of $I_J^{(N)}  \mathcal{H}^{K_N}_{J}  $  are given by
\begin{align*}
 \langle k |  D-   I_J^{(N)}  D_N \overline I_J^{(N)}   |  {k^\prime}\rangle = & \ \langle k | L_y^2  |  {k^\prime}\rangle -  J_N^{-1} \langle J - k | S_y^2  | J - {k^\prime}\rangle  + \omega^2 \left[ \langle k | L_x^2  |  {k^\prime}\rangle - J_N^{-1} \langle J - k | S_x^2 | J - {k^\prime}\rangle  \right]
\end{align*}
Each of the  two terms terms in the right side are explicit  thanks to~\eqref{def:L} and~\eqref{eq:MatrixSorth}.  
If $ k^\prime  \leq K_N-2 $, their difference can be expressed in terms of 
$$
 \langle m | L_x | m^\prime \rangle -  J_N^{-1/2} \langle J - m | S_x  | J - {m^\prime}\rangle = \delta_{|m-m^\prime|, 1} \left[ \sqrt{\frac{\max\{ m, m^\prime\} }{2}} -\sqrt{\frac{2J \max\{ m, m^\prime\}- m m^\prime\ }{4J_N }} \right] 
$$
with $ m, m^\prime \leq K_N $ and analogously  for $ y $ instead of $ x $, for which the expression  differs only by an overall complex phase. By an explicit calculation, the modulus of the last term is upper bounded according to
$$
\max_{|J-J_N| \leq \overline{K}_N} \left|  \langle m | L_\xi | m^\prime \rangle -  J_N^{-1/2} \langle J - m | S_\xi  | J - {m^\prime}\rangle \right| \leq   \frac{ \max\{ m, m^\prime, \overline{K}_N\}^{3/2}   }{J_N \sqrt{ 1- \max\{ m, m^\prime, \overline{K}_N\}/J_N}} \  \delta_{|m-m^\prime|, 1} 
$$
for both $ \xi \in \{ x, y\} $. 
Since also 
\begin{align*}
 \left|  \langle m | L_\xi | m^\prime \rangle \right|  \leq \sqrt{\max\{ m, m^\prime\}/2} &\leq  \sqrt{K_N}, \\
   J_N^{-1/2} \left|   \langle J-  m |  S_\xi  | J- m^\prime \rangle \right| & \leq \sqrt{K_N/|\mathbf{m}_0|} ,
\end{align*} 
we conclude that for some constant $ C< \infty $ and all $ k\leq K_N  $, $ k^\prime \leq K_N-2 $:
\begin{equation}
\max_{|J-J_N| \leq \overline{K}_N}  \left| \langle k |  D-    I_J^{(N)}  D_N \overline I_J^{(N)}  |  {k^\prime}\rangle  \right| \leq  \frac{C}{N}  \max\big\{ K_N^2 ,\sqrt{K_N \overline{K}_N^3}\big\} \ 1[|k-k^\prime| \leq 2 ] .
\end{equation}
By an analogous estimate as in~\eqref{eq:normfrom kernel}, we thus arrive at the claimed norm estimates in~\eqref{eq:normDDN}.\\

\noindent
2.~Since $ c_{00} $ is a common core for $ D $ and   $\overline D_N = P_{K_N} D P_{K_N}  $, the claimed strong-resolvent convergence is immediate from from the fact that the matrix elements of $ D $ vanish unless their difference is smaller than two (cf.~\cite{RS80}). 
To boost this strong convergence  to convergence of eigenvalues, we use  the block approximation
\begin{equation}\label{eq:Dblock} 
\widehat D_N \coloneqq P_{K_N} D P_{K_N} +  Q_{K_N} D Q_{K_N} , \qquad \mbox{on $ \ell^2(\mathbb{N}_0) $.} 
\end{equation}
with $ Q_{K_N} = 1 - P_{K_N} $. For $ E_K = \sum_{k=0}^{K-1} |\psi_k\rangle\langle \psi_k| $, we have
 $$
\left\| E_K Q_{K_N} \right\|^2 = \left\| E_K Q_{K_N} E_K \right\| \leq \sum_{k=0}^{K-1} \sum_{n = K_N}^\infty | \langle n | \psi_k \rangle |^2 = o(K_N^{-\infty}) .  
$$
The error estimate as $ N \to \infty $ follows from the exponential decay estimate~\eqref{eq:expdecay}. 

The matrix elements of $\langle n | D-\widehat D_N | m \rangle $  are non-vanishing only for $ | n-K_N|, |m-K_N| \leq 2 $. Moreover, $ | \langle n | D- \widehat D_N | m \rangle | \leq \mathcal{O}(K_N) $ such that again by the exponential decay estimate~\eqref{eq:expdecay}:
\begin{equation}\label{eq:projectionbd}
\left\| (D- \widehat D_N) E_K \right\| = o(K_N^{-\infty}) .
\end{equation}
Denoting by $ F_K = 1 - E_K $, this implies that $ \left\| E_K (D- \widehat D_N) E_K \right\| = o(K_N^{-\infty})  $ and $ \left\| F_K \widehat D_N E_K \right\| = o(K_N^{-\infty})  $. Since also
\begin{align*}
F_K  \widehat D_N F_K & \geq F_K P_{K_N} D P_{K_N} F_K + E_{K}(D) \, F_K Q_{K_N} F_K Q_{K_N} F_K \\
& = F_K P_{K_N} \left(D - E_{K}(D) \right) P_{K_N} F_K + E_{K}(D) \,  \left( F_K  - F_K Q_{K_N} E_K Q_{K_N} F_K \right) \\
	& \geq E_{K}(D) \ F_K \left( 1 - 2 \| Q_{K_N} E_K \|^2 \right) = E_{K}(D) \ F_K \left( 1 - o(K_N^{-\infty}) \right) ,
\end{align*}
the Schur-complement method of Proposition~\ref{prop:Schur} proves that the eigenvalues of $\widehat D_N$ strictly below $ E_{K}(D)  \left( 1 - o(1) \right)  $ aymptotically coincide with those of $ D $ up to an error of order $ o(K_N^{-\infty})  $. 

A similar estimate  also shows that $ Q_{K_N} D Q_{K_N} \geq E_{K}(D) Q_{K_N}  \left( 1 - o(1) \right)  $ for any $ K\in \mathbb{N}$, so that the low-energy spectrum of $\widehat D_N$ entirely coincides with that of its first block $  \overline D_N $. This finishes the proof of~\eqref{eq:convev}. 

The assertion~\eqref{eq:convspecprN} then follows from~\eqref{eq:projectionbd}, the discrete nature of the spectrum of $ D $ and a standard perturbation theory bound based on the representation of the spectral projection as a contour integral involving resolvents (cf.~\cite{RS80}).
%
\end{proof}

\subsection{Proof of Theorem~\ref{thm:unimin}}\label{sec:pfunim}
For a proof of Theorem~\ref{thm:unimin}, we analyse all the blocks $ H_{J,\alpha} $ in the decomposition~\eqref{eq:Hblock} of the Hamiltonian separately.  The main contribution to the lowest energies of $ H $ will come from the largest $ J $ and the vicinity of the minimum. Projection techniques will help to focus on this patch. 

To facilitate notation, by a unitary rotation 
we subsequently assume without loss of generality that the minimum of $ h $ is at $  \mathbf{m}_0 =  (0,0, 1)^T $, and that 
the projected Hessian $ Q_\perp D_h(\mathbf{m}_0) Q_\perp $ as a matrix on $ \ran Q_\perp $ has its eigenvectors aligned with the $ x $- and $ y $-direction with $ \omega_x $ and $ \omega_y $ the corresponding eigenvalues. This
entails that 
\begin{equation}\label{eq:rotQ}
\widehat Q_N( \mathbf{m}_0) =   N h( \mathbf{m}_0) \mathbbm{1} + \left(N - 2 S_z  \right) | \nabla h( \mathbf{m}_0)|  + \frac2{ N}   \left( \omega_x S_x^2 + \omega_y S_y^2\right)  ,
\end{equation}
cf.~\eqref{eq:defQhat}. We will also assume without loss of generality that  the positive eigenvalues of $ D_h^\perp(\mathbf{m}_0) $ are ordered according to 
$
 0 < \omega_y^\perp =  \omega_y + |\nabla h(\mathbf{m}_0) | \leq  \omega_x + |\nabla h(\mathbf{m}_0) | =  \omega_x^\perp $, 
and we note that
\begin{equation}\label{eq:omega}
\sqrt{\det   D_P^{\perp}(\mathbf{m}_0) }  = \sqrt{\omega_x^\perp \omega_y^\perp}  =  \omega_y^\perp  \, \omega , \quad \mbox{with $ \displaystyle \omega^2 := \frac{ \omega_x^\perp}{\omega_y^\perp} \geq 1 $.} 
\end{equation}
Moreover, we will use throughout the whole proof $ K_N = \overline{K}_N = o(N^{1/3}) $ diverging as $ N \to \infty $.

\subsubsection{Limit operator for $ J \geq N/2 - K_N $}

By assumption~\eqref{ass:quadr} and Lemma~\ref{lem:altQ}, on the increasing subspaces $ \mathcal{H}^{K_N}_{J}( \mathbf{m}_0)  $ with $  J \geq N/2 - K_N $, the shifted quadratic polynomial $ \kappa( \mathbf{m}_0) \mathbbm{1} + \widehat Q_N( \mathbf{m}_0) $ approximates $ H_{J,\alpha} $ uniformly in norm to order $ o(1) $.  Next, we show  that this quadratic polynomial is  well approximated by
$$
	H_J^{(N)} := N h( \mathbf{m}_0)  + \kappa( \mathbf{m}_0) +  | \nabla h( \mathbf{m}_0)|  (N-2J-1) +   \omega_y^\perp \ \overline I_J^{(N)} D   I_J^{(N)}  \quad \mbox{on $  \mathcal{H}^{K_N}_{J}( \mathbf{m}_0)  $}
$$
with $ D = \omega^2 L_x^2+ L_y^2  $ on $ \ell^2(\mathbb{N}_0) $ and $ \omega $ from~\eqref{eq:omega}.
\begin{lemma}\label{lem:approxH}
In the situation of Theorem~\ref{thm:unimin} 
if  $ K_N = o(N^{1/3}) $:
\begin{equation}\label{eq:approxH}
\max_{ J \geq N/2 - K_N } \max_{\alpha} \left\| \left( H_{J,\alpha} - H_J^{(N)}\right) P_{J}^{K_N^-}(\mathbf{m}_0)   \right\| = o( 1) .
\end{equation}
\end{lemma}
\begin{proof}
We use assumption~\eqref{ass:quadr} and Lemma~\ref{lem:altQ}  to establish the claim~\eqref{eq:approxH}  with $ H_J^{(N)} $ replaced by 
$ \kappa( \mathbf{m}_0) \mathbbm{1} +  \widehat Q_N( \mathbf{m}_0)  $. 
In order to simplify the term in~\eqref{eq:rotQ} proportional to $  | \nabla h( \mathbf{m}_0)|  $, we rewrite 
$$
N - 2 S_z 
= \frac{N}{2} - \frac{2}{N} \left(\mathbf{S}^2 - S_x^2 - S_y^2 \right)   + \frac{(N-2S_z)^2}{2N} .
$$
On $  \mathcal{H}^{K_N}_{J}( \mathbf{m}_0)  $ with $ J \geq  N/2 - K_N $, the norm of the last term is bounded by $ (4 K_N)^2/2N 
= o(1) $. 
To that order, we can therefore replace 
$ H_J^{(N)} $   by 
\begin{align*} 
\widehat  H_J^{(N)} := N h( \mathbf{m}_0) +  \kappa( \mathbf{m}_0)  +  | \nabla h( \mathbf{m}_0)| \left( N -2 J  - 1\right) +  \omega^\perp_y D_N   , & \\
	\mbox{with}\quad  D_N = \frac{\omega^2}{J_N}  S_x^2 + \frac{1}{J_N}  S_y^2 \quad  \mbox{and}\quad &   J_N = \frac{N}{2}.  
\end{align*}
Lemma~\ref{lem:src} then yields
$$ \max_{J \geq N/2-K_N} \left\|(  H_J^{(N)}  -  \widehat H_J^{(N)}) P_{J}^{K_N^-}(\mathbf{m}_0)   \right\| = o( 1) ,
$$
which completes the proof.
\end{proof}

\subsubsection{Truncations} 

To control the blocks $ H_{J,\alpha} $  for values of $ J $ other than the ones of the previous paragraphs, we use the following lemma.
\begin{lemma}\label{lem:truncation}
In the situation of Theorem~\ref{thm:unimin}, there are constants $ c, C \in (0,\infty) $ such that for all $ (J,\alpha) $:
\begin{equation}
H_{J,\alpha} \geq  N h(\mathbf{m}_0) - C + c \left( N -  2  \ \mathbf{e}_{  \mathbf{m}_0} \cdot \mathbf{S} \right) .
\end{equation}
\end{lemma}
\begin{proof}
From assumption~\eqref{ass:symbol} we infer that for some $ C < \infty $:
\begin{equation}\label{eq:fromasssym}
H_{J,\alpha } \geq - C +   \frac{2J+1}{4\pi} \int  N h\Big(\frac{2J}{N}  \mathbf{e}(\Omega) \Big) \,  \big| \Omega, J \rangle \langle \Omega, J \big| \  d\Omega . 
\end{equation}
Since $   \mathbf{m}_0 \in S^2 $ was assumed to be the unique minimum of $ h: B_1 \to \mathbb{R} $ and $ |\nabla h(   \mathbf{m}_0 ) | > 0 $, there is some $ c > 0 $ such that for all $  \mathbf{m} \in B_1 $:
\begin{equation}\label{eq:lowerh1min}
h( \mathbf{m}) \geq  h( \mathbf{m}_0) + c (1 - \mathbf{e}_{  \mathbf{m}_0} \cdot \mathbf{m} ) . 
\end{equation}
Plugging this estimate into the above operator inequality and using the fact that  
by Lemma~\ref{lem:assz}
$$
 \frac{2J+1}{4\pi} \int d\Omega \,  J  \mathbf{m}_0 \cdot \mathbf{e}(\Omega)   \big| \Omega, J \rangle \langle \Omega, J \big| \ d\Omega \geq - C + \mathbf{e}_{  \mathbf{m}_0} \cdot \mathbf{S} 
$$
with some constant $ C $, which does not depend on $ N, J $, we arrive at the claimed matrix inequality. 
\end{proof}

\subsubsection{Finishing the proof}

With the above preparations, we are ready to spell out the proof of Theorem~\ref{thm:unimin}. Aside from exploiting the block decomposition~\eqref{eq:Hblock}, the argument essentially relies on a two-step approximation procedure and a Schur-complement analysis in the main blocks corresponding to $ J \geq N/2 - K_N $. 
\begin{proof}[Proof of Theorem~\ref{thm:unimin}]
We investigate the spectrum of each block $ H_{J,\alpha} $ of the Hamiltonian in~\eqref{eq:Hblock} separately, and distinguish cases.\\

In case $ J \leq  N/2 - K_N $, we use Lemma~\ref{lem:truncation} to conclude that for all $ \alpha $:
\begin{equation}
 H_{J,\alpha} \geq  N h(\mathbf{m}_0) - C + 2 c K_N .  
\end{equation}
Since $ K_N \to \infty $ as $ N \to \infty $, these blocks clearly do not contribute to the asserted low-energy spectrum below 
$ E_0(H) + \Delta $ with arbitrary $ \Delta  > 0 $.\\

In case $ J \geq  N/2 - K_N $, we consider the approximating Hamiltonian $H_J^{(N)} $ on $  \mathcal{H}^{K_N}_{J}( \mathbf{m}_0) $ from Lemma~\ref{lem:approxH}, and define its projection
$$
\widetilde H_J^{(N)} :=  P_{J}^{K_N^-}(\mathbf{m}_0) H_J^{(N)}   P_{J}^{K_N^-}(\mathbf{m}_0)   \quad \mbox{on $ \mathcal{H}^{K_N^-}_{J}( \mathbf{m}_0)  $.}
$$ 
Note that this operator is unitarily equivalent to
\begin{equation}\label{eq:DNcalc}
 \left( N h( \mathbf{m}_0) \mathbbm{1} + \kappa( \mathbf{m}_0) +  | \nabla h( \mathbf{m}_0)|  (N-2J-1) \right)  P_{K_N^-}  +   \omega_y^\perp P_{K_N^-} D P_{K_N^-} \; \mbox{on $P_{K_N^-} \ell^2(\mathbb{N}_0) $.}
\end{equation}
We fix $ K \in \mathbb{N} $ arbitrary, and let $ E_K^{(N)} $ stand for the orthogonal projection onto the $ K $-dimensional subspace of $ \mathcal{H}^{K_N^-}_{J}( \mathbf{m}_0) $  spanned by eigenvectors  of the $ K $  lowest eigenvalues of $ \widetilde H_J^{(N)} $.  Its orthogonal complement in $ \mathbbm{C}^{2J+1} $ will be denoted by $ F_K^{(N)}  = \mathbbm{1}_{\mathbb{C}^{2J+1}} - E_K^{(N)}  $. 
Using Lemma~\ref{lem:approxH} we conclude
\begin{align}\label{eq:mainblock}
	& \max_{ J \geq   N/2 - K_N } \max_{\alpha} \left\|  \widetilde H_J^{(N)}  - E_K^{(N)}  H_{J,\alpha} E_K^{(N)}  \right\| = o(1)  , \\
	& \max_{ J \geq   N/2 - K_N } \max_{\alpha}   \left\|  F_K^{(N)}  H_{J,\alpha}  E_K^{(N)} \right\| \leq  \max_{ J \geq   N/2 - K_N } \left\|  F_K^{(N)}  \overline D_{J,N} E_K^{(N)} \right\|  +  o(1) . \notag 
\end{align}
The term in the right side equals
\begin{equation}\label{eq:FEblockweg}
\left\|  F_K^{(N)}   \overline I_J^{(N)} ( D -P_{K_N^-} D P_{K_N^-}) I_J^{(N)}  E_K^{(N)} \right\|  \leq \left\| ( D -P_{K_N^-} D P_{K_N^-}) I_J^{(N)}  E_K^{(N)} \right\| = o(1) ,
\end{equation}
and its convergence follows from~\eqref{eq:convspecprN} and~\eqref{eq:projectionbd}. 

For the Schur-complement principle in Proposition~\ref{prop:Schur} it thus remains to control the lowest eigenvalue of the block $F_K^{(N)}  H_{J,\alpha}  F_K^{(N)} $. For this task, we again employ Lemma~\ref{lem:truncation}, which yields the matrix inequality
\begin{align}\label{eq:QHQ}
F_K^{(N)}  H_{J,\alpha}  F_K^{(N)}   & \geq \left( N h(\mathbf{m}_0) - C\right) F_K^{(N)}    + c\  F_K^{(N)}  \left( N - 2 S_z \right) F_K^{(N)}   \notag \\ & \geq \left( N h(\mathbf{m}_0) - C\right) F_K^{(N)}    + 2 c \, M \left( 1  - \left\| F_K^{(N)}  P^M_J  F_K^{(N)}   \right\| \right) F_K^{(N)} 
\end{align}
where $ P^M_J  $ is the orthogonal projection in $ \mathbb{C}^{2J+1} $  spanned by the eigenvectors $ | J \rangle , | J-1 \rangle , \dots | J - M \rangle $ of $ S_z $ corresponding to the highest $ M+1 $ eigenvalues. We establish a lower bound on the last term in~\eqref{eq:QHQ} with the help of an upper bound on
\begin{equation}\label{eq:turnaround}
\| F_K^{(N)}  P^M_J  F_K^{(N)}  \| = \| P^M_J F_K^{(N)}  P^M_J \|  \leq \sum_{m=0}^M \langle J-m | \  F_K^{(N)}   | J-m\rangle . 
\end{equation}
By Lemma~\ref{lem:src} as $ N \to \infty $ and for any sequence $ J \geq N/2 - K_N $ the matrices $P_{K_N^-} D P_{K_N^-}   $ converge in strong-resolvent sense to $ D $.  This implies the weak convergence of the corresponding spectral projections. In particular,  for each fixed $ m \in \mathbb{N}_0 $ and any sequence $ J \geq N/2 - K_N $:
\begin{equation}\label{eq:estconf} 
  \langle J-m | \  F_K^{(N)}  | J-m\rangle = \sum_{k > K} \left| \langle m | \psi_k \rangle \right|^2 + o(1) , 
\end{equation}
where $ (\psi_k)_{k\in \mathbb{N}_0} $ denotes the orthonormal eigenbasis of $ D $. The latter are estimated in Proposition~\ref{prop:harm}. Thanks to the exponential decay estimate~\eqref{eq:expdecay2}, the right side is exponentially small in $ K - M $ for all $ m \in \{ 0, 1, \dots , M \} $ provided $ K \geq 2 M $. 
Hence, if we choose $ M $ large enough and subsequently $ K $ suitably larger,
the prefactor in the  last term in~\eqref{eq:QHQ} exceeds any constant. The block $ F_K^{(N)} H_{J,\alpha}  F_K^{(N)}  $ then does not contribute to the low-energy spectrum of $ H_{J,\alpha}  $.

By~\eqref{eq:mainblock}--\eqref{eq:estconf} and Proposition~\ref{prop:Schur} the low-energy spectrum of $ H_{J,\alpha} $ hence asymptotically agrees with the lowest eigenvalues of $\widetilde H_J^{(N)} $. In turn, they agree with those of the operator in~\eqref{eq:DNcalc}.
 By the convergence~\eqref{eq:convev}, the eigenvalues of the last term in~\eqref{eq:DNcalc} asymptotically agree with those of  $\omega_y^\perp  D $, which are given by the oscillator values~\eqref{eq:harmosc}. For $ k = N/2 - J \in \mathbb{N}_0 $ fixed, this yields the expression~\eqref{eq:allev} for the  eigenvalues including their multiplicities. 
 
The minimal eigenvalue $ E_0(H) $ asymptotically corresponds to the choice $ k = m = 0 $, which establishes~\eqref{eq:gsenergy}. 
The corresponding weak convergence of the eigenvector~\eqref{eq:eigenvect} is a consequence of the Schur-complement analysis,  Lemma~\ref{lem:src} and the explicit expression~\eqref{eq:gs} for the ground-state of $ D $. 

Finally, the expression~\eqref{eq:gapH} for the spectral gap is immediate from~\eqref{eq:allev}, since either the first excited corresponds to $ k= 0 $ and $ m = 1 $ or vice versa.  
\end{proof}

\subsection{Proof of Theorem~\ref{thm:unimin2}}\label{sec:pfunim2}
In case the unique minimum is found at at  $  \mathbf{m}_0 \in B_1 $ with $0 < | \mathbf{m}_0| < 1$, the low-energy spectrum of $ H $ stems again from blocks $ H_{J,\alpha} $ with $ J $-values in the vicinity of $ J_N( \mathbf{m}_0) = N  |\mathbf{m}_0|/ 2 $. This is evident from the following lower bound, which substitutes Lemma~\ref{lem:truncation} in the present case. 
\begin{lemma}\label{lem:truncation1a}
In the situation of Theorem~\ref{thm:unimin2}, there are constants $ c, C \in (0,\infty) $ such that for all $ (J,\alpha) $ and all $ \mathbf{m}_1 \in B_1 $:
\begin{equation}
H_{J,\alpha} \geq  N h(\mathbf{m}_0) - C - \frac{cN}{4} |    \mathbf{m}_1- \mathbf{m}_0 |^2 + \frac{c}{N} \left( J - J_N( \mathbf{m}_1)  \right)^2 + c | \mathbf{m}_1|   \left( J -    \mathbf{e}_{  \mathbf{m}_1} \cdot \mathbf{S} \right) .
\end{equation}
\end{lemma}
\begin{proof}
We use that $   \mathbf{m}_0 \in B_1$ is the unique minimum of $ h: B_1 \to \mathbb{R} $ at which $ D_h(  \mathbf{m}_0) > 0 $. Hence there is some $ c > 0 $ such that:
\begin{equation}\label{eq:simplelb}
h( \mathbf{m}) \geq  h( \mathbf{m}_0) + \frac{c}{4} \left|  \mathbf{m}- \mathbf{m}_0 \right|^2 \geq h( \mathbf{m}_0) + \frac{c}{8} \left|  \mathbf{m}- \mathbf{m}_1 \right|^2 -  \frac{c}{2} \left|  \mathbf{m}_1 - \mathbf{m}_0 \right|^2 .
\end{equation}
The last inequality was obtained using Cauchy-Schwarz and holds  for all $  \mathbf{m} \in B_1 $. 
Plugging the estimate into~\eqref{eq:fromasssym} yields the claim using the same arguments as in the proof of Lemma~\ref{lem:truncation}. 
\end{proof}

\begin{proof}[Proof of Theorem~\ref{thm:unimin2}] 
A Rayleigh-Ritz bound with the variational state $ | \Omega_0 , J_N(\mathbf{m}_0) \rangle $, with $ \Omega_0 $ the spherical angle of $  \mathbf{e}_{\mathbf{m}_0 }$, yields the upper bound 
$$
E_0(H) \leq E_0\big(H_{J_N(\mathbf{m}_0) ,\alpha}\big) \leq \langle \Omega_0 , J_N(\mathbf{m}_0)| H_{J_N(\mathbf{m}_0) ,\alpha}  | \Omega_0 , J_N(\mathbf{m}_0) \rangle  \leq Nh(\mathbf{m}_0) + C 
$$
by Proposition~\ref{prop:Druff2}. Combing this with Lemma~\ref{lem:truncation1a} with $ \mathbf{m}_1 = \mathbf{m}_0 $, we immediately conclude that the blocks $ H_{J,\alpha} $  
corresponding to  $ | J - J_N( \mathbf{m}_0) | \geq C  \sqrt{N }  $ with $ C > 0 $ suitably large, but fixed, do not contribute to the spectrum near the ground state.

It thus remains to analyse the blocks corresponding to  $ | J - J_N( \mathbf{m}_0) | \leq  C \sqrt{N }  $.  
For any $ J $ in this regime, we associate a radius $ r_{J,N} :=  2J/N $, which by construction satisfies
$ \left| r_{J,N} - |\mathbf{m}_0| \right| \leq C N^{-1/2}   $. The value of $ h $  at the point on this sphere in direction of $ \mathbf{m}_0 $ is controlled by a second-order Taylor estimate using $ | \nabla h(\mathbf{m}_0) | = 0 $:
\begin{equation}\label{eq:Taylorh}
\left| h\left(  r_{J,N} \mathbf{e}_{\mathbf{m}_0} \right) - h(\mathbf{m}_0) \right| \leq C  \left| r_{J,N} - |\mathbf{m}_0| \right|^2 \leq \frac{C}{N} ,
\end{equation}
where here and in the following the constant $ C $ changes from line to line, but remains independent of $ J $ and $ N $. 
Hence, the minimum $\mathbf{m}_{J,N} := \arg\min \{ h(\mathbf{m}) \, | \, |\mathbf{m} | = r_{J,N} \} $ on the sphere of radius $  r_{J,N}  $ has distance to $ \mathbf{m}_0 $ at most
\begin{equation}\label{eq:distancems}
	\left| \mathbf{m}_{J,N}- \mathbf{m}_0 \right| \leq  \frac{C}{\sqrt{N}} , 
\end{equation}
since otherwise the first inequality in~\eqref{eq:simplelb} together with the above estimate of $ h\left(  r_{J,N} \mathbf{e}_{\mathbf{m}_0} \right) $ would yield a contradiction to the minimality of $ h(\mathbf{m}_{J,N} ) $. Using Taylor-estimates for $ h \in C^3 $, this also implies:
\begin{align}\label{eq:Taylor3}
	& \left|\nabla h(\mathbf{m}_{J,N} ) \right| = 	\left|\nabla h(\mathbf{m}_{J,N} ) - \nabla h( \mathbf{m}_0)\right| \leq C \left| \mathbf{m}_{J,N}- \mathbf{m}_0 \right| \leq  \frac{C}{\sqrt{N}}  , \notag \\
	& | \kappa(\mathbf{m}_{J,N} ) - \kappa( \mathbf{m}_0)| \leq C \left| \mathbf{m}_{J,N}- \mathbf{m}_0 \right| \leq  \frac{C}{\sqrt{N}}  , \notag \\
	& \left\| D_h(\mathbf{m}_{J,N} ) -  D_h( \mathbf{m}_0 ) \right\| \leq C \left| \mathbf{m}_{J,N}- \mathbf{m}_0 \right| \leq   \frac{C}{\sqrt{N}} .
\end{align}
Having singled out $ \mathbf{m}_{J,N}  $, we now use assumption \textbf{A1} to approximate $ H_{J,\alpha} $ in the $ \mathcal{H}_{J}^{K_N}(\mathbf{m}_{J,N} ) $, where 
we pick $ K_N = o(N^{1/6}) $ diverging to infinity.  This enables us to use  Lemma~\ref{lem:altQ} (with $ \overline{K}_N = 1 $) and the  second-order polynomial
$ \widehat Q(\mathbf{m}_{J,N}) $ in this approximation. The above estimate on the gradient implies that the first-order term is negligible:
$$
\left\| \nabla h(\mathbf{m}_{J,N} ) \cdot \left( 2\mathbf{S} -N\mathbf{m}_{J,N} \right)  P_{J}^{K_N}(\mathbf{m}_{J,N})  \right\| 
\leq  2 \left|  \nabla h\mathbf{m}_{J,N} ) \right|   K_N  \leq C \frac{K_N}{\sqrt{N}} = o(1)  . 
$$
Changing coordinates by a unitary rotation such that $ \mathbf{m}_{J,N} = (0,0,r_{J,N}) $, we thus conclude that on $ \mathcal{H}_{J}^{K_N}( \mathbf{m}_{J,N} )  $ we may approximate $ H_{J,\alpha} $  in terms of 
$$
N h(\mathbf{m}_{J,N}  ) +  \kappa(\mathbf{m}_{J,N} ) + \frac{2 \ r_{J,N}}{N} \left( \omega_x^{(N)} S_x^2 +  \omega_y^{(N)}  S_y^2 \right) . 
 $$
 where  $\omega_x^{(N)} ,  \omega_y^{(N)} $ are the two eigenvalues of $ D_h^\perp(\mathbf{m}_{J,N} ) $, i.e., the Hessian projected onto the directions perpendicular to $ \mathbf{m}_{J,N} $. Since $  D_h( \mathbf{m}_0 )  > 0 $, these eigenvalues are uniformly bounded away from zero for all $ | J - J_N( \mathbf{m}_0) | \leq  C \sqrt{N }  $ by~\eqref{eq:Taylor3}. Moreover, they uniformly converge to the eigenvalues $   0 < \omega_y \leq \omega_x$ of $ D_h^\perp(\mathbf{m}_{0} ) $ as $ N \to \infty $. In fact, from~\eqref{eq:Taylor3} we have the estimates $ | \omega_\xi^{(N)} -  \omega_\xi | \leq C N^{-1/4} $ for both $ \xi \in \{x,y\} $. Since also $\left\| S_\xi^2 P_{J}^{K_N}(\mathbf{m}_{J,N})  \right\| \leq C N K_N $ (cf.~\eqref{eq:dthpowerS}), 
we may thus use
$$
H_J^{(N)} := N h(\mathbf{m}_{J,N}  ) +  \kappa(\mathbf{m}_{0} )  +  | \mathbf{m}_{0} | \,  \omega_y \ \overline I_J^{(N)} \left( \omega^2 L_x^2 + L_y^2\right)  I_J^{(N)}  \quad \mbox{on $  \mathcal{H}^{K_N}_{J}( \mathbf{m}_{J,N})  $}
$$
to show the following substitute for Lemma~\ref{lem:approxH}:
$$
\max_{ | J - J_N( \mathbf{m}_0) | \leq  C \sqrt{N } } \left\| \left( H_{J,\alpha} - H_J^{(N)} \right) P_{J}^{K_N}(\mathbf{m}_{J,N}) \right\| = o(1) . 
$$
We then proceed with the Schur-complement analysis as in the proof of Theorem~\ref{thm:unimin} with only one modification.  To control the last block, we use Lemma~\ref{lem:truncation1a} with $ \mathbf{m}_1 = \mathbf{m}_{J,N} $. 
For this choice, the quadratic difference term in its right side is bounded by a constant thanks to~\eqref{eq:distancems}. %

This proves that $ E_0(H_{J,\alpha}) =  N h(\mathbf{m}_{J,N})  + \kappa( \mathbf{m}_0) +  | \mathbf{m}_0|  \sqrt{\det   D_h^{\perp}(\mathbf{m}_0)} + o(1) $ for all blocks $ (J,\alpha) $ with $  | J - J_N( \mathbf{m}_0) | \leq  C \sqrt{N } $. Clearly for $ J = J_N( \mathbf{m}_0) $, we have $ \mathbf{m}_{J,N} = \mathbf{m}_0 $, which concludes the proof of~\eqref{eq:gsball}. The assertion concerning the regime $  | J - J_N( \mathbf{m}_0) | \leq o(\sqrt{N})  $ follows by a Taylor estimate as in~\eqref{eq:Taylorh}. 
\end{proof}

\subsection{Proof of Theorem~\ref{thm:mmin}}\label{sec:prmmin}
Compared to the case of one minimum, the case of several minima $ \mathcal{M} = \{ \mathbf{m}_1, \dots ,  \mathbf{m}_L \} \subset S^2 $ of $ h $ poses the additional problem to separate the patches around them. In the next subsection, we use semiclassical analysis based on refined projection techniques  to tackle this issue and to focus on the subspaces, from which the low-energy spectrum arises. The part going beyond this semiclassicss largely parallels the proof in the previous section.
\subsubsection{Semiclassics for subspace decompositions} 
In each of the  subspaces $ \mathcal{H}_J^{K}(\mathbf{m}_l) \subset \mathbb{C}^{2J+1} $, which are associated to the minimizing directions
$ \mathbf{m}_l $,  we may choose the canonical orthonormal basis consisting of normalized eigenvectors  of $ \mathbf{m}_l \cdot \mathbf{S} $, i.e.
$$
\mathbf{m}_l \cdot \mathbf{S}  \ | k ; \mathbf{m}_l \rangle =  k \ | k ; \mathbf{m}_l \rangle , \quad k \in \{-J, \dots, J \} , 
$$
and hence $  \mathcal{H}_J^{K}(\mathbf{m}_l) = \spa\left\{ | J- k ; \mathbf{m}_l \rangle \ | \ k \in \{ 0, 1 , \dots  , K \} \right\} $. 
\begin{lemma}\label{lem:spectrumG}
Let $ \mathcal{M} := \{ \mathbf{m}_1, \dots , \mathbf{m}_L \} \subset S^2 $ be a finite set and $ J \in  \mathbb{N} /2 $ and  $ K \in \mathbb{N} $ fixed. Then the spectrum of the Gram matrix
\begin{equation}\label{def:G}
G := \left( \langle J -k^\prime  ; \mathbf{m}_{l^\prime}    | J -k  ; \mathbf{m}_l \rangle  \right)_{\substack {l,l^\prime \in \{1, \dots , L \} \\  k,k^\prime \in \{0,1, \dots , K \} } } 
\end{equation}
is contained in the set $ [1-R_J^K , 1+ R_J^K ] $ with 
$$
R_J^K = (L-1) (K+1) ( 4K J)^{K}  \exp\left(- (J-2K) \gamma \right) \quad \mbox{and}\quad \gamma = \min_{l\neq l\prime}  \ln   \left(\cos\frac{ \sphericalangle( \mathbf{m}_l, \mathbf{m}_{l^\prime}) }{2}\right)^{-1} > 0  , 
$$
with $  \sphericalangle( \mathbf{m}_l, \mathbf{m}_{l^\prime}) $ the spherical angle between the two points. 
\end{lemma}
\begin{proof}
From Lemma~\ref{lem:overlap} we infer that for all $ k,k^\prime \leq K $:
\begin{equation}\label{eq:bounded}
\left| \langle J- k ; \mathbf{m}_l  \big|   J-k^\prime ; \mathbf{m}_{l^\prime} \rangle  \right|  \leq ( 2J)^{(k+k^\prime)/2} (2\max\{k,k^\prime\})^{\min\{k,k^\prime\}}  \left| \cos\frac{ \sphericalangle( \mathbf{m}_l, \mathbf{m}_{l^\prime}) }{2} \right|^{2J-(k+k^\prime)} .
\end{equation}
Since $ G $ is a hermitian block matrix with $ L \times L $ blocks of size $ K+1 $ and the diagonal blocks all equal to the unit matrix, the assertion follows straightforwardly from Gershgorin's circle theorem. 
\end{proof}
The last lemma ensures that for $ J \geq N/2 - K_N $ and $ K_N = o(N /\ln N )  $: 
\begin{equation}\label{eq:Rexpsmall}
R_J^{K_N}= o(N^{-\infty}) .
\end{equation}
The union of the individual basis vectors of $  \mathcal{H}_J^{K_N}(\mathbf{m}_l) $ for any finite number of directions hence still form a set of linearly independent vectors and thus a basis of joint subspace
$$
 \mathcal{H}_J^{K_N} := \bigvee_{l=1}^L  \mathcal{H}_J^{K_N}(\mathbf{m}_l) . 
$$
In this situation, we may construct an orthonormal basis using the square-root of the inverse, $G^{-1/2} $, of  the Gram matrix defined as in~\eqref{def:G}:
\begin{equation}\label{def:ONB} 
	\big| (k,l) \rangle := \sum_{l^\prime =1}^L \sum_{k^\prime =0}^{K_N} G^{-1/2}_{(kl),(k^\prime l^\prime)} \big|  J-k^\prime ; \mathbf{m}_{l^\prime} \rangle  , \qquad (k,l) \in \{ 0, \dots , K_N\} \times \{ 1, \dots , L \} . 
\end{equation}
Clearly, $\langle (k^\prime,l^\prime) \big| (k,l) \rangle = \delta_{l, l^\prime} \delta_{k, k^\prime} $.  The spectral projections onto $ \mathcal{H}_J^{K_N}(\mathbf{m}_l)  $ and its cousins after orthogonalization, 
\begin{equation}\label{eq:defstraightproj}
P_J^{K_N}( \mathbf{m}_l) := \sum_{k=0}^{K_N}  \big|  J -k  ; \mathbf{m}_l  \rangle \langle  J -k  ; \mathbf{m}_l \rangle  \big| , \qquad P_J^{K_N}(l) :=  \sum_{k=0}^{K_N}  \big|  (k,l)\rangle  \langle (k,l) \rangle  \big| ,
\end{equation}
are then norm-close to each other.
\begin{theorem}\label{cor:approxP}
For $ N/2 \geq J \geq N/2 - K_N $ and $ K_N = o(N/ \ln N )  $:
\begin{enumerate}
\item $ \displaystyle 
	\left\| \big| (k,l) \rangle - \big|  J-k ; \mathbf{m}_{l} \rangle \right\| = o(N^{-\infty}) 
$ for all $ (k,l) \in \{0,\dots, K_N\} \times \{1,\dots, L\} $.
\item 
$\displaystyle 
\max_{l} \left\| P_J^{K_N}( \mathbf{m}_l) -  P_J^{K_N}(l)  \right\| = o(N^{-\infty}) $, and 
the projection $ P_J^{K_N} := \sum_{l=0}^L P_J^{K_N}(l) $ onto $  \mathcal{H}_J^{K_N}  $ is approximated according to
\begin{equation}\label{eq:almostdecompunity}
\left\| P_J^{K_N}- \sum_{l=0}^LP_J^{K_N}( \mathbf{m}_l)  \right\| = o(N^{-\infty}) . 
\end{equation}
\end{enumerate}
\end{theorem}
\begin{proof}
1.~By definition we have 
$$
 \big|  J-k ; \mathbf{m}_{l} \rangle = \sum_{l^\prime =1}^L \sum_{k^\prime =0}^{K_N} G^{1/2}_{(kl),(k^\prime l^\prime)} \big| (k^\prime,l^\prime) \rangle  .
 $$
The norm of the difference vector is hence bounded according to
$$
\left\| \big| (k,l) \rangle - \big|  J-k ; \mathbf{m}_{l} \rangle \right\| \leq \| G^{1/2} - \mathbbm{1} \| \leq \max_{| \lambda -1 | \leq R_J^{K_N} } \left|  \sqrt{\lambda} -1 \right| \leq \frac{R_J^{K_N}}{2 \sqrt{1-R_J^{K_N}}} . 
$$ 
The last inequality follows from the bounds on the eigenvalues of $ G $ established in Lemma~\ref{lem:spectrumG}. The assertion is therefore a simple consequence of~\eqref{eq:Rexpsmall}. 

\noindent
2.~The assertions follow immediately from the first item with the help of the  triangle inequality. 
\end{proof}

In order to be able to control the relation of vectors in $  \mathcal{H}_J^{K_N}(\mathbf{m}_l) $ to a slightly fattened version of $   \mathcal{H}_J^{K_N}(\mathbf{m}_{l^\prime}) $ with $ l\neq l^\prime$,  we also need the following lemma. It will play a crucial role in the truncation procedure below.
\begin{lemma}\label{lem:onefar}
Let 
\begin{equation}\label{ass:kappa}
 0 < \kappa < \frac{1}{2} \left( \sin\frac{ \sphericalangle( \mathbf{m}_l, \mathbf{m}_{l^\prime}) }{2}\right)^2 ,
 \end{equation}
 and suppose that 
$ N/2 \geq J \geq N/2 - K_N $ and $ K_N = o(N/ \ln N )  $. Then for any $ k \in \{0, \dots , K_N\} $:
\begin{equation}
\sum_{k=0}^{K_N} \sum_{k^\prime =0}^{\kappa N} \left| \langle J -k^\prime  ; \mathbf{m}_{l^\prime}    | J -k  ; \mathbf{m}_l \rangle\right|^2 = o(N^{-\infty} ) .  
\end{equation}
\end{lemma} 
\begin{proof}
We split the $ k^\prime $-sum in two terms. The first sum to $ K_N $ is estimates with the help of~\eqref{eq:bounded} from which the claim follows by the same lines of reasoning as above. For the second sum, we have $ k \leq k^\prime $. 
Abbreviating $ \theta :=  \sphericalangle( \mathbf{m}_l, \mathbf{m}_{l^\prime}) $, we estimate this part with the help of~\eqref{eq:binomest} by
$$
\sum_{k=0}^{K_N}\binom{2J}{k} \left( \kappa N \right)^{2k} \left( \frac{2 + 2 (\sin\frac{\theta}{2})^2}{\sin\theta} \right)^{2k} \left(\cos\frac{\theta}{2}\right)^{4J} \sum_{k^\prime=0}^{\kappa N} \binom{2J}{k^\prime} \left( \tan\frac{\theta}{2}\right)^{2k^\prime} .
$$
The truncated binomial is bounded by a standard Chernoff bound for any $ t > 0 $
\begin{equation}
\left(\cos\frac{\theta}{2}\right)^{4J}  \sum_{k^\prime=0}^{\kappa N} \binom{2J}{k^\prime} \left( \tan\frac{\theta}{2}\right)^{2k^\prime}  \leq  e^{t\kappa N} \left( 1- (1- e^{-t} ) \  \left( \sin\frac{\theta}{2}\right)^{2}\right)^{N} . 
\end{equation} 
Choosing $ t(\theta) = \ln \frac{1 - \kappa}{\kappa} \left(\tan\frac{\theta}{2}\right)^2 > 0 $, the right side is of the form $ \exp\left( -N \alpha(\theta) \right) $ with $ \alpha(\theta) =  - \kappa t(\theta) - \ln\left[ \cos^2\left(\frac{\theta}{2}\right)/(1-\kappa) \right]< 0 $.
Since the remaining summation is estimated trivially by the number of terms $ K_N $ times the maximum of each term, which occurs at $ k = K_N $ with the binomial also trivially bounded, $ \binom{2J}{k}  \leq N^k \leq N^{K_N} $, the result follows. 
\end{proof}
\subsubsection{Truncation}

By assumption, the minima of $ h $  have the property $h(\mathbf{m}) \geq  h( \mathbf{m}_1 ) = \dots = h( \mathbf{m}_L ) $ for all $\mathbf{m} \in B_1 $.  A substitute for~\eqref{eq:lowerh1min} is provided by a bound of the form
\begin{equation}\label{eq:lowerhLmin}
 h(\mathbf{m}) \geq  h( \mathbf{m}_1 ) + f_0 +  \sum_{l=1}^L f\left(\mathbf{m}_l\cdot \mathbf{m}\right) 
\end{equation}
with $ f_0 > 0 $ and a monotone decreasing $ C^2$-function $ f: [0,1] \to [0,\infty) $ of the form
$$
f(x) =\begin{cases} 0 & \mbox{if} \quad 0\leq x \leq \xi,  \\
	c(1-x) - f_0 &  \mbox{if} \; \frac{1+\xi}{2}  \leq x \leq 1 ,
\end{cases} 
$$
with some  $ c > 0 $. The parameter $ \xi \in (0,1) $ is chosen close enough to one such that  the supports corresponding to distinct $ l \neq l^\prime $ do not overlap, i.e., 
$
f(\mathbf{m}_l\cdot \mathbf{m})  f(\mathbf{m}_{l^\prime}\cdot \mathbf{m}) = 0 $. We will choose
\begin{equation}\label{eq:defxi}
0 < 1- \xi < \min_{l \neq l^\prime} \left( \sin\frac{ \sphericalangle( \mathbf{m}_l, \mathbf{m}_{l^\prime}) }{2}\right)^2 .
\end{equation}

The same strategy as in the proof of Lemma~\ref{lem:truncation} immediately yields that for some constant $ C_L \in (0,\infty) $ and all $ (J,\alpha) $:
\begin{equation}\label{eq:oldtrunc}
H_{J,\alpha} \geq  N h(\mathbf{m}_1) - C_L + N f_0 +  \sum_{l=1}^L N  f\left(\frac{2}{N} \mathbf{m}_l\cdot \mathbf{S}\right)   .
\end{equation}
This would enable us to discard all blocks $ J < N/2 - K_N $ as far the low-energy spectrum is concerned. For the other blocks, we however need a slightly more refined lower bound which distinguishes-patches corresponding to the subspace decomposition in the previous subsection.

\begin{lemma}\label{lem:truncation2}
In the situation of Theorem~\ref{thm:mmin}, there are constants $ C , c \in (0,\infty) $ such that for all $ (J,\alpha) $ with $ J \geq N/2- K_N $ and $  K_N = o(N/ln N) $:
\begin{equation}\label{eq:truncation3}
H_{J,\alpha} \geq  N h(\mathbf{m}_1) - C + c K_N  Q_J^{K_N} + \sum_{l=1}^L c \left( N - 2 \mathbf{m}_l\cdot \mathbf{S} \right) P_J^{K_N}(\mathbf{m}_l)  ,
\end{equation}
where $Q_J^{K_N} := \mathbbm{1}_{\mathbb{C}^{2J+1} } -  P_J^{K_N} $ is the orthogonal projection to the complement of $ \mathcal{H}_J^{K_N} $.
\end{lemma}
\begin{proof}
We start from~\eqref{eq:oldtrunc}. In order to ease the notation,  in this proof we abbreviate $   \hat f_l := N f\left(\frac{2}{N} \mathbf{m}_l\cdot \mathbf{S}\right) $ and  we will drop the super-/subscripts on the projection, e.g. $ P := P_J^{K_N} $. We will estimate the block of this operator in the decomposition $ P+Q = \mathbbm{1} $ separately. For the blocks involving $ P $, we use
\begin{equation}\label{eq:approx1}
\left\|  \hat f_l P -  \hat f_l P( \mathbf{m}_l)  \right\|  \leq \| \hat f_l  \| \left\| P- \sum_{l=0}^LP( \mathbf{m}_l)  \right\| 
+   \sum_{l^\prime \neq l } \left\|  \hat f_l P( \mathbf{m}_{l^\prime})  \right\|
\end{equation}

The operator $   \hat f_l $ is diagonal in the eigenbasis of $  \mathbf{m}_l\cdot \mathbf{S}\ $, i.e.
\begin{equation}\label{eq:spectralf}
 \hat f_l  = \sum_{k=0}^{\lfloor \frac{N}{2}(1-\xi) \rfloor}  N f\left(\frac{2}{N} (J-k) \right) \ \big| J-k; \mathbf{m}_l\rangle\langle J-k; \mathbf{m}_l \big| ,
\end{equation}
where the truncation of the $ k $-sum results from the bounds on the support of $ f $. Evidently $ \| \hat f_l  \| \leq N f_0 $. Therefore, for any $ l^\prime \neq l $:
$$
 \left\|  \hat f_l P( \mathbf{m}_{l^\prime})  \right\| \leq N f_0 \left( \sum_{k=0}^{\lfloor \frac{N}{2}(1-\xi) \rfloor} \sum_{k^\prime=0}^{K_N} \left| \langle J-k; \mathbf{m}_l \big| J-k^\prime; \mathbf{m}_{l^\prime} \rangle \right|^2 \right)^{1/2}  =o(N^{-\infty}) 
$$
where the last step is Lemma~\ref{lem:onefar}. Its applicability is ensured by the choice~\eqref{eq:defxi} of $ \xi $. From~\eqref{eq:approx1} and Theorem~\ref{cor:approxP}, we this conclude $ \left\|  \hat f_l P -  \hat f_l P( \mathbf{m}_l)  \right\| = o(N^{-\infty}) $. 
By the spectral representation~\eqref{eq:spectralf}  we also have for all sufficiently large $ N $:
$$
 \hat f_l P( \mathbf{m}_l) =  \left[ c N \left( 1 - \frac{2}{N} \mathbf{m}_l\cdot \mathbf{S} \right) - N f_0\right] P( \mathbf{m}_l) .
$$
Upon summation over $ l \in \{ 1, \dots , L \} $ and adding $ Nf_0 P $, this term produces the last term in the right side of~\eqref{eq:truncation3} up to another norm-error of order $ o(N^{-\infty}) $ due to~\eqref{eq:almostdecompunity}. These error terms are absorbed in the constant in~\eqref{eq:truncation3}.

It thus remains to investigate the block $ Q  \hat f_l  Q + N( f_0+ h(\mathbf{m}_1)) Q $.  To do so, it is most convenient to switch back to the representation using coherent states. Since $ f \in C^2 $ this may be done at the expense of another constant thanks to~\eqref{prop:Druff2}.
We then lower bound
$$
Q \left[\frac{2J+1}{4\pi}\int N h\left(\frac{2}{N} \mathbf{e}(\Omega) \right) \ | \Omega ,J \rangle \langle \Omega , J | \ d\Omega \right]  Q \geq \frac{2J+1}{4\pi} \int_{C_N^c }  N h\left(\frac{2}{N} \mathbf{e}(\Omega) \right) Q | \Omega ,J \rangle \langle \Omega , J | Q \ d\Omega ,
$$
where $ C_N^c = S^2 \backslash \bigcup_{l=1}^L C_N(\mathbf{m}_l)   $ is the complement of the union of the spherical caps 
$$ C_N(\mathbf{m}_l) := \left\{ \Omega \ | \   \mathbf{e}(\Omega) \cdot  \mathbf{m}_l \geq 1 - \frac{K_N}{2N} \right\} . $$
These are chosen such that on $ C_N^c $ we have the lower bound $ N  h(\mathbf{m}) \geq N h(\mathbf{m}_1) + c K_N $. Thanks to the decomposition of unity~\eqref{eq:completeness}, to complete the proof, it  remains to establish an upper bound for all $ l $ on:
\begin{align*}
\left\| \frac{2J+1}{4\pi} \int_{ C_N(\mathbf{m}_l)   } \mkern-20mu Q | \Omega ,J \rangle \langle \Omega , J | Q \ d\Omega \right\| & \leq  \frac{N+1}{4\pi } \int_{ C_N(\mathbf{m}_l)   } \mkern-20mu \left\| Q | \Omega ,J \rangle \right\|^2 d \Omega \\
& 
\end{align*}
The spherical volume of $ C_N(\mathbf{m}_l) $ is $ \pi K_N /N $. To estimate the norm in the integrand, we  fix $ \Omega \in C_N(\mathbf{m}_l) $ and employ the approximate decomposition of unity as expressed in~\eqref{eq:almostdecompunity}:
$$
\left\| Q | \Omega ,J \rangle \right\| \leq \left\| (\mathbbm{1} - P( \mathbf{m}_l) ) | \Omega ,J \rangle \right\| + \sum_{l^\prime \neq l} \left\|  P( \mathbf{m}_{l^\prime})  | \Omega ,J \rangle \right\| + o(N^{-\infty}) . 
$$
By a unitary rotation, we may assume without loss of generality that $  \mathbf{m}_l = (0,0,1)^T $. In this case, an estimate on the first term is contained in Proposition~\ref{lem:chernoff} in the appendix. For its application we note that $ 2J \sin^2(\frac{\theta}{2}) = J ( 1-\cos \theta) \leq K_N /4 $. Choosing $ \delta = K_N / [ 8J \sin^2(\frac{\theta}{2})] \geq 1 $, we hence conclude:
$$
\left\| (\mathbbm{1} - P( \mathbf{m}_l) ) | \Omega ,J \rangle \right\|^2  \leq \sum_{k \geq K_N/2 } \left| \langle J-k | \Omega, J\rangle \right|^2 \leq \exp\left( - \frac{K_N}{12} \right) = o(K_N^{-\infty}) . 
$$
As a consequence of this, we also have for all $ l^\prime \neq l $
$$
\left\|  P( \mathbf{m}_{l^\prime})  | \Omega ,J \rangle \right\|  \leq \left\|  (\mathbbm{1} - P( \mathbf{m}_l) ) | \Omega ,J \rangle \right\| + \| P( \mathbf{m}_{l^\prime}) P( \mathbf{m}_l) \| \leq   o(K_N^{-\infty})  +  o(N^{-\infty}) ,
$$
where the last estimate is due to Theorem~\ref{cor:approxP}. This completes the proof. 
\end{proof}

\subsubsection{Finishing the proof} 

Once the above semiclassical reduction to the relevant subspaces in the vicinity of the minima  is accomplished, the proof largely follows the strategy of the case of one minimum. In the following. we will therefore only highlight the differences. \\

We start by setting some notation.  By assumption the projection $  D_h^{\perp}(\mathbf{m}_l) $ of the Hessian of $ h $ at each minimum onto the plane perpendicular to $ \mathbf{m}_l $ has two strictly positive eigenvalues, 
$$
0 < \omega_{y,l}^\perp =  \omega_y + |\nabla h( \mathbf{m}_l)|\leq  \omega_{x,l}^\perp = \omega_x + |\nabla h( \mathbf{m}_l)| , \qquad \mbox{and we set}\quad \omega_l^2 := \frac{\omega_{x,l}^\perp}{\omega_{y,l}^\perp} \geq 1 . 
$$
Throughout the  proof we again use $ K_N = \overline{K}_N = o(N^{1/3}) $ diverging as $ N \to \infty $. 

By assumption~\eqref{ass:quadr} and Lemma~\ref{lem:altQ}, on the increasing subspaces $ \mathcal{H}^{K_N}_{J}( \mathbf{m}_l)  $, we approximate $ H_{J,\alpha} $ in terms of
the quadratic term $ \widehat Q(\mathbf{m}_l) $, which involves $ 2 \left( \omega_{x,l} S_x^2 + \omega_{y,l}  S_y^2 \right)/N $.
Proceeding as in Lemma~\ref{lem:approxH}, we therefore arrive at
\begin{equation}\label{eq:approxH1}
	\left\| \left( H_{J,\alpha} -H_J^{(N)}(l)  \right) P_J^{K_N^-}(\mathbf{m}_l) \right\| = o(1) 
\end{equation}
where
$$
H_{J,l}^{(N)} \coloneqq N h( \mathbf{m}_l)  + \kappa( \mathbf{m}_l) +  | \nabla h( \mathbf{m}_l)|  (N-2J-1) +  \overline I_{J,l}^{(N)}  \omega_{y,l}^\perp  D(l) I_{J,l}^{(N)}     \quad \mbox{on $  \mathcal{H}^{K_N}_{J}( \mathbf{m}_l)  $.}
$$
The operator $ D(l) := \omega_l^2 L_x^2 + L_y^2 $ acts in $ \ell^2(\mathbb{N}_0) $ and 
$$
I_{J,l}^{(N)} : \mathcal{H}_J^{K_N}(\mathbf{m}_l)  \to \ell^2(\mathbb{N}_0) , \qquad \overline I_{J.l}^{(N)} : \ell^2(\mathbb{N}_0) \to  \mathcal{H}_J^{K_N}(\mathbf{m}_l) 
$$
is the natural injection respectively projection with respect to the $z $-basis in the $ l $-direction, i.e.  $  I_{J,l}^{(N)} | J-k;\mathbf{m}_l \rangle = | k \rangle $ for all $ k \in \{0, \dots , K_N\} $. Theorem~\ref{cor:approxP} ensures the proximity of these $z $-basis vectors to the orthonormalized basis $ | (k,l) \rangle $, which compose an orthonormal basis for the joint subspace $ \mathcal{H}_J^{K_N} = \bigvee_{l=1}^L  \mathcal{H}_J^{K_N}(\mathbf{m}_l) $. We therefore replace the above isometric embeddings and projections by
\begin{align*} I_J^{(N)}: \;  \mathcal{H}_J^{K_N}  \to \bigoplus_{l=1}^L \ell^2(\mathbb{N}_0) , \quad  \mbox{and}\quad  \overline I_J^{(N)} : \bigoplus_{l=1}^L \ell^2(\mathbb{N}_0) \to   \mathcal{H}_J^{K_N}  ,
\end{align*}
where $ I_J^{(N)} | (k,l) \rangle = | k \rangle_l $. 
These embeddings are direct sums,  $  I_J^{(N)} = \bigoplus_{l=1}^L I_J^{(N)}(l) $, of embeddings of $\mathcal{H}_{J}^{K_N}(l)  \coloneqq P_J^{K_N}(l) \mathbb{C}^{2J+1} $  into the $l $th copy of $ \ell^2(\mathbb{N}_0) $,  cf.~\eqref{eq:defstraightproj}.
Theorem~\ref{cor:approxP} then allows us to replace $ H_{J,l}^{(N)}  $ on $  \mathcal{H}^{K_N}_{J}( \mathbf{m}_l)  $ by
$$
   H_J^{(N)}(l) \coloneqq N h( \mathbf{m}_l)  + \kappa( \mathbf{m}_l) +  | \nabla h( \mathbf{m}_0)|  (N-2J-1) +  \overline I_{J}^{(N)}(l) \omega_{y,l}^\perp  D(l) I_{J}^{(N)}(l)    \quad \mbox{on $  \mathcal{H}^{K_N}_{J}(l)  $.}
$$ 
These operators can be lifted to the direct sum
$$
H^{(N)}_J \coloneqq \bigoplus_{l=1}^L    H_J^{(N)}(l) \qquad \mbox{on}\quad \mathcal{H}_{J}^{K_N} =  \bigoplus_{l=1}^L \mathcal{H}_{J}^{K_N}(l)   .
$$
The above argument, then yields to following modification of Lemma~\ref{lem:approxH}. 
\begin{lemma}\label{lem:approxH2}
In the situation of Theorem~\ref{thm:mmin} 
if  $  K_N = o(N^{1/3}) $:
\begin{equation}
\max_{ J \geq N/2 - K_N } \max_{\alpha} \left\| \left( H_{J,\alpha} - H_J^{(N)}\right) P_{J}^{K_N^-}  \right\| = o( 1) ,
\end{equation}
where $ \displaystyle P_{J}^{K_N^-} \coloneqq \sum_{l=1}^L \sum_{k=0}^{K_N^-} | (k,l)\rangle\langle (k,l) | $ with $ K_N^- = K_N -2 $.
\end{lemma}
The proof is straightforward from~\eqref{eq:approxH1} and Theorem~\ref{cor:approxP}. Equipped with this, we then proceed with the proof of Theorem~\ref{thm:mmin} in the same way as for the case of one minimum. 

\begin{proof}[Proof of Theorem~\ref{thm:mmin}]

In case $ J \leq N/2 - K_N $ we use~\eqref{eq:oldtrunc} to conclude that the ground-state  of $ H_{J.\alpha} $ is found above $ Nh(\mathbf{m}_1) - C + c K_N $ and hence does not contribute at the energies considered.\\

In case $ J > N/2 - K_N $ we consider the projections,
$$
\widetilde H^{(N)}_J \coloneqq P_{J}^{K_N^-} H^{(N)}_J P_{J}^{K_N^-} \qquad \mbox{on $ P_{J}^{K_N^-} \mathcal{H}_{J}^{K_N} $.}
$$
Note that this matrix is still a direct sum of matrices associated with the subspaces corresponding to $ P_{J}^{K_N^-}(l)  = \sum_{k=0}^{K_N^-} | (k,l)\rangle\langle (k,l) | $. 
The matrices forming the direct sum are 
unitarily equivalent to 
$$
\left( N h( \mathbf{m}_l)  + \kappa( \mathbf{m}_l) +  | \nabla h( \mathbf{m}_0)|  (N-2J-1) \right)  P_{K_N^-} +  \omega_{y,l}^\perp \ P_{K_N^-} D(l)   P_{K_N^-} 
$$
on $  P_{K_N^-} \ell^2(\mathbb{N}_0) = \spa\left\{ | k\rangle | k \in \{0, \dots , K_N^-\} \right\} $. In turn, these operators have been described in detail in Section~\ref{sec:limop}. 

We now
fix $ K \in \mathbb{N} $ arbitrary, and let $ E_{K}^{(N)} $ stand for the orthogonal projection onto the subspace of $  P_{J}^{K_N^-} \mathcal{H}_{J}^{K_N}  $  spanned by eigenvectors  of the $ K $  lowest eigenvalues of $ \widetilde H_J^{(N)} $.  Its orthogonal complement in $ \mathbbm{C}^{2J+1} $ will be denoted by $ F_{K}^{(N)}  = \mathbbm{1}_{\mathbb{C}^{2J+1}} - E_{K}^{(N)}  $. 

Lemma~\ref{lem:approxH2} ensures the validity of the estimates~\eqref{eq:mainblock}--\eqref{eq:FEblockweg} (with minor modifications in the notation). It thus remains to again control the block $ F_{K}^{(N)}   H_{J,\alpha} F_{K}^{(N)} $.  To do so, we modify the argument in~\eqref{eq:QHQ}. With the help of Lemma~\ref{lem:truncation2}, we arrive at:
\begin{equation}\label{eq:QHQ2}
F_{K}^{(N)}   H_{J,\alpha} F_{K}^{(N)}  \geq  \left( N h(\mathbf{m}_1) - C \right) F_{K}^{(N)} + c \min\{ 2 M , K_N \}  F_{K}^{(N)} - 2 c M   F_{K}^{(N)}  \sum_{l=1}^L  \left\| F_{K}^{(N)}  P_J^{M}(\mathbf{m}_l)   F_{K}^{(N)}  \right\| ,
\end{equation}
where $ M \in \mathbb{N} $ is arbitrary. Using Theorem~\ref{cor:approxP} we can replace the projection $  P_J^{M}(\mathbf{m}_l)  $ by $  P_J^{M}(l) $ at the expense of a term which is $ o(N^{-\infty}) $. The latter can be added to the order one term proportional to $ C $. Proceeding as in~\eqref{eq:turnaround}, it thus remains to estimate
$$
\sum_{l=1}^L \left\|  P_J^{M}(l)   F_{K}^{(N)}  P_J^{M}(l)   \right\| \leq \sum_{l=1}^L \sum_{m=0}^M \langle (m,l) | F_{K}^{(N)} | (m,l) \rangle . 
$$
Since  $ \widetilde H^{(N)}_J $ is a  direct sum,  for each of the terms in the $ l $-sum we are therefore back to \eqref{eq:estconf} with $ K $ changed depending on how the $ L $ harmonic oscillator levels interlace. Since  \eqref{eq:estconf} was identified to be exponentially small in case $ K $ is chosen much larger than $ M $, this still shows that the last term in the right side of \eqref{eq:QHQ2} is bounded independent of $ M $. Hence choosing $ M $ large enough and subsequently $ K $ larger, the ground state energy of the block $ F_{K}^{(N)}   H_{J,\alpha} F_{K}^{(N)} $ is seen to be much larger than the energies of interest. 

By a Schur-complement analysis (Proposition~\ref{prop:Schur}) the low-energy spectrum of $ H_{J,\alpha} $ agrees with that of $ \widetilde H^{(N)}_J  $.  In the limit of $ N \to \infty $ and using Lemma~\ref{lem:src}, the spectrum of $ \widetilde H^{(N)}_J  $ is a direct sum of $ L $ harmonic oscillator spectra as described in Proposition~\ref{prop:harm}. 
\end{proof}

\appendix
\section{Miscellanea on  spin-coherent states}\label{app:cs}
In this appendix, we collect properties of the spin coherent states as defined in~\eqref{def:cs}. We restrict attention to the semiclassical properties, which were essential for the analysis in this paper. We refer to the textbooks \cite{Per86,Gaz09,CR12} and \cite{ACGT72,Lieb73} for further information and references. 
\subsection{Semiclassical estimates for the states}\label{app:cs1}
The spin coherent states~\eqref{def:cs} on $ \mathbb{C}^{2J+1} $ are parametrised by the angles $ \Omega =(\theta,\varphi) $ on the unit sphere. Their scalar product 
$$
\langle \Omega^\prime, J | \Omega , J\rangle = \left[ \cos  \frac{\theta }{2}  \cos  \frac{\theta^\prime }{2}+ e^{i(\varphi-\varphi')}  \sin \frac{\theta }{2}  \sin  \frac{\theta^\prime }{2}\right|^{2J}
$$
shows that for large values of $ J \in \mathbb{N}/2 $, they are sharply localised. Denoting by $ \sphericalangle(\Omega,\Omega^\prime) $ the spherical angle between two points on the unit sphere, one has the Gaussian-type localisation 
$$
\left| \langle \Omega^\prime, J | \Omega , J\rangle\right|^2 = \left[ \cos  \sphericalangle(\Omega,\Omega^\prime) \right]^{4J}
$$
with width proportial to $ J^{-1/2} $.

With respect to the orthonormal eigenbasis of $ S_z $ on $ \mathbb{C}^{2J+1} $, the spin coherent states are linear combinations with coefficients given by
\begin{equation}\label{eq:repcoh}
\langle J-k | \Omega , J\rangle =  \binom{2J}{k}^{1/2}\left( \cos  \frac{\theta }{2} \right)^{2J-k}\ \left( \sin \frac{\theta }{2} \right)^{k} e^{ik \varphi} ,   
\end{equation}
for any $ k \in \{0,1, \dots, 2J\}  $ and $ \Omega =(\theta,\varphi)$. 
Measurement of $ S_z $ will therefore result in a binomial distribution of $ 2J $ independent Bernoulli variables with parameter $ p = \sin^2\left(\frac{\theta}{2}\right) $. The following lemma records the
standard upper-tail Chernoff estimate for the binomial distribution.
\begin{proposition}\label{lem:chernoff}
For any $ J \in \mathbb{N} /2 $ and any $ k \in \{0,1, \dots, 2J\} $ and any $ \delta > 0 $:
\begin{equation}
\sum_{k \geq 2 (1+\delta) J\sin^2\left(\frac{\theta}{2}\right)} \left|  \langle J-k | \Omega , J\rangle \right|^2 \leq \exp\left(- \frac{\delta^2}{2+\delta} 2J \sin^2\left(\frac{\theta}{2}\right) \right)  . 
\end{equation}
\end{proposition}
We will also need the following generalisation of the identity~\eqref{eq:repcoh}, which involves the unitary $ U(\Omega) $ defined in~\eqref{def:cs}.
\begin{lemma}\label{lem:overlap}
For any $ k , k' \in \{0,1, \dots, 2J\} $ and any $\Omega =(\theta,\varphi) $:
\begin{align}\label{eq:overlap}
\langle J-k^\prime | U(\Omega)  | J-k \rangle = \mkern-5mu \sum_{m=\max\{0, k-k^\prime \} }^k  \mkern-25mu  (-1)^k \ & \binom{2J+m-k}{m}^{1/2}\binom{2J+m-k}{k^\prime -k + m}^{1/2} 
	  \binom{k}{m}^{1/2} \binom{k^\prime}{k^\prime-k+m}^{1/2} \notag \\ & \times  \left( \cos  \frac{\theta }{2} \right)^{2J-k-k^\prime}\ \left( \sin \frac{\theta }{2} \right)^{2m +k^\prime-k}  e^{i (k^\prime-k)\varphi} . 
\end{align}
Moreover, if $ k\leq k^\prime $ (with the convention that $ 0^0 = 1 $):
\begin{equation}\label{eq:binomest}
\left| \langle J-k^\prime | U(\Omega)  | J-k \rangle  \right| \leq   \binom{2J}{k}^{1/2} \binom{2J}{k'}^{1/2}(k^\prime)^k  \left( \cos  \frac{\theta }{2} \right)^{2J-k-k^\prime} \left( \sin \frac{\theta }{2} \right)^{k^\prime-k}  \left( 1+ \left(\sin \frac{\theta }{2} \right)^{2}\right)^k
\end{equation}
\end{lemma}
\begin{proof}
We use the decomposition of the unitary~\cite[Eq.~(4.3.14)]{Per86} (see also~\cite{ACGT72}):
$$
U(\Omega) =  \exp\left( \tan \frac{\theta}{2} \  e^{i\varphi} S_- \right) \exp\left( 2 \ln\left( \cos\frac{\theta}{2} \right) S_z \right) \exp\left( - \tan \frac{\theta}{2} e^{-i\varphi} S_+ \right)  . 
$$
The formula~\eqref{eq:overlap} follows from a straightforward, but tedious calculation which uses that
$$
\frac{1}{m!} \left( S_+\right)^m | J - k \rangle = \binom{k}{m}^{1/2}  \binom{2J +m-k }{m}^{1/2}   | J - k + m \rangle 
$$
and similarly for $ S_- $ (and the usual convention that the binomial is zero if the upper integer is smaller than the lower one), cf.~\cite[Eq.~(4.2.3)]{Per86}.

The bound~\eqref{eq:binomest} then follows by estimating
\begin{align*}
\binom{2J+m-k}{k^\prime -k + m}  \binom{k^\prime}{k^\prime-k+m} & =  \binom{2J}{k^\prime} \frac{(2J +m-k)!}{(2J)!(k-m)!} \left(\frac{(k^\prime)!}{(k^\prime+m-k)!} \right)^2 \\
& \leq \binom{2J}{k^\prime}  \frac{(k^\prime)^{2k}}{(k-m)!} ,
\end{align*}
and similarly 
$$
 \binom{2J+m-k}{m} \leq  \binom{2J}{k} \frac{k!}{m!} . 
$$
This yields the claim by binomial formula. 
\end{proof}

\subsection{Semiclassical estimates for the symbols}
Associated to any linear operator on $ \mathbb{C}^{2J+1} $ are two semiclassical symbols: the lower and upper one. We recall from~\cite{Lieb73} the lower symbols of the spin operator
\begin{equation}\label{eq:upperS}
\langle \Omega, J | \mathbf{S}  |  \Omega , J\rangle = J \mathbf{e}(\Omega) . 
\end{equation} 
as well as its upper symbol alongside, the symbol of the square of the $ z $-component,
\begin{align}\label{eq:symbolS}
	\mathbf{S} &= \frac{2J+1}{4\pi} \int d\Omega \,  (J+1) \mathbf{e}(\Omega)  \, \big| \Omega, J \rangle \langle \Omega, J \big| \\
	S_z^2 &=  \frac{2J+1}{4\pi} \int d\Omega \, ( (J+1)(J+3/2) \mathbf{e}_z(\Omega)^2 -(J+1)/2 )  \, \big| \Omega, J \rangle \langle \Omega, J \big| .  \label{eq:symbolS32}
\end{align}

For operators with a smooth upper symbol, the  lower symbol is know to agree with the upper symbol in the semiclassical limit~\cite{Lieb73,Duffield:1990aa}. Since we need quantitative error estimates, we include the following result, which is tailored to the applications considered here (see also \cite[Prop.~4.2]{Landsman:2020aa}). 

\begin{proposition}\label{prop:Druff2}
For any 
$$
H = \frac{2J+1}{4\pi}  \int  d\Omega  \, N f \Big(\frac{2J}{N}  \mathbf{e}(\Omega) \Big) \,  \big| \Omega, J \rangle \langle \Omega, J \big| \, 
$$
on $ \mathbb{C}^{2J+1} $ with  some $ f \in C^2 $, there is some $ C < \infty $ such that for all $ J \leq N/2 $: 
\begin{equation}
\sup_{\Omega} \left|  \langle \Omega, J \big|  H \big| \Omega, J \rangle - N f \big(\frac{2J}{N}  \mathbf{e}(\Omega) \big) \right| \leq C . 
\end{equation}
\end{proposition}
\begin{proof}
The proof is based on the standard Taylor estimate
\begin{align*}
 f \Big(\frac{2J}{N}  \mathbf{e}(\Omega^\prime) \Big)  =  f \Big(\frac{2J}{N}  \mathbf{e}(\Omega) \Big)  + \frac{2J}{N}   \nabla  f \Big(\frac{2J}{N}  \mathbf{e}(\Omega) \Big)  \cdot \left(  \mathbf{e}(\Omega^\prime) -  \mathbf{e}(\Omega) \right) + r(\Omega,\Omega';J/N) & \\ 
\mbox{with} \quad  \left| r(\Omega,\Omega';J/N)\right| \leq \| f'' \|_\infty \left( \frac{2J}{N} \right)^2  \left( 1-  \mathbf{e}(\Omega^\prime) \cdot \mathbf{e}(\Omega) \right) . & 
 \end{align*} 
Using the representation~\eqref{eq:symbolS} and~\eqref{eq:upperS} and the completeness~\eqref{eq:completeness}, we arrive at
\begin{align*} 
& \frac{2J+1}{4\pi}  \int d\Omega^\prime  \,   \nabla  f \Big(\frac{2J}{N}  \mathbf{e}(\Omega) \Big)  \cdot \left(  \mathbf{e}(\Omega^\prime) -  \mathbf{e}(\Omega) \right) \left| \langle \Omega, J | \Omega^\prime , J \rangle \right|^2 \\ & =  \nabla  f \Big(\frac{2J}{N}  \mathbf{e}(\Omega) \Big)  \cdot \left( \frac{1}{J+1} \langle \Omega, J | \mathbf{S}  |  \Omega , J\rangle  - \mathbf{e}(\Omega)  \right)  =  \nabla  f \Big(\frac{2J}{N}  \mathbf{e}(\Omega) \Big)  \cdot \mathbf{e}(\Omega)  \left( \frac{J}{J+1} -1\right) .
\end{align*} 
The remainder term is estimated similarly
\begin{equation*}
	\frac{2J+1}{4\pi}  \int d\Omega^\prime   \left| r(\Omega,\Omega';J/N)\right|  \left| \langle \Omega, J | \Omega^\prime , J \rangle \right|^2 \leq  \| f'' \|_\infty \left( \frac{2J}{N} \right)^2  \left( 1 - \frac{J}{J+1}\right) . 
\end{equation*} 
Inserting the Taylor estimate into the integral expression for $ H $ and using the normalization $ \langle \Omega, J | \Omega , J\rangle = 1 $, we thus arrive at
\begin{align*}
& \left|  \langle \Omega, J \big|  H \big| \Omega, J \rangle  - N f \big(\frac{2J}{N}  \mathbf{e}(\Omega)\big)  - 2J \  \nabla  f \Big(\frac{2J}{N}  \mathbf{e}(\Omega) \Big) \cdot   \mathbf{e}(\Omega)  \left( \frac{J}{J+1} -1\right) \right| \\
&\quad \leq \| f'' \|_\infty  \frac{4J^2}{N} \left( 1 - \frac{J}{J+1}\right)
 \end{align*} 
from which the claim follows.  
\end{proof}

The last proposition shows the consistency of the two symbols in the semiclassical limit. We also need the following quantitative version of Duffield's theorem \cite{Duffield:1990aa} on the consistency of the quantisation of non-commuting self-adjoint polynomials with the help of the lower symbol. 

\begin{proposition}\label{prop:Druff1}
	If $ H = N \ \textrm{P}\Big(\tfrac{2}{N}   \mathbf{S} \Big)$ on $ \mathbb{C}^{2J+1} $ with a non-commuting self-adjoint polynomial $P$, then there is some $ C < \infty $, which is independent of $ J \leq N/2 $, such that 
	\begin{equation}
		\left\| H  -     \frac{2J+1}{4\pi} \int d\Omega \, N P \Big(\frac{2J}{N}  \mathbf{e}(\Omega) \Big) \,  \big| \Omega, J \rangle \langle \Omega, J \big|  \right\| \leq C . 
	\end{equation}
\end{proposition}
The proof, which is spelled out at the end of this subsection, will rest on two preparatory lemmas. The first lemma deals with operators of the type $f(2S_z/N)$.
\begin{lemma}\label{lem:assz}
	If $ H = N f(2 \mathbf{v} \cdot \mathbf{S} /N)$ on $ \mathbb{C}^{2J+1} $ for some $f \in C^2$, then there is some $ C < \infty $, which is independent of $ J \leq N/2 $, such that  for all unit vectors $ \mathbf{v} \in S^2 $ and $ N $, $ J \leq N/2 $:
	\begin{equation}
		\left\| H  -    \frac{2J+1}{4\pi} \int N f \Big(\frac{2J}{N} \mathbf{v} \cdot  \mathbf{e}(\Omega) \Big) \,  \big| \Omega, J \rangle \langle \Omega, J \big|  \,  d\Omega \right\| \leq C .
	\end{equation}
\end{lemma}

\begin{proof}
	We first reduce the assertion to the case $  \mathbf{v} = \mathbf{e}_z $.
	If $\Omega$ stands for the spherical angles of $\mathbf{v} =  \mathbf{e}(\Omega) $, 
we have
$ U(\Omega)^*   (\mathbf{v}\cdot    \mathbf{S} )U(\Omega) = S_z $ with the unitary from~\eqref{def:cs}. 
Similarly, by the definition of the coherent states, one easily arrives at 
$$ U(\Omega_0)^*  \left(   \int N f \Big(\frac{2J}{N}    \mathbf{v}\cdot   \mathbf{e}(\Omega)  \Big) \,  \big| \Omega, J \rangle \langle \Omega, J \big| \, d\Omega\right) U(\Omega_0)  = \int  \, N f \Big(\frac{2J}{N}  \mathbf{e}_z(\Omega) \Big) \,  \big| \Omega, J \rangle \langle \Omega, J \big| 
d\Omega $$
for any continuous function $f$. 

	In order to establish the claim in case $  \mathbf{v} = \mathbf{e}_z  $, we first show that 
	\[ H^\prime \coloneqq   \frac{2J+1}{4\pi} \int d\Omega \, N f \Big(\frac{2J}{N}  \mathbf{e}_z(\Omega) \Big) \,  \big| \Omega, J \rangle \langle \Omega, J \big|  \]
	s diagonal in the orthonormal basis $|k \rangle$ with $k \in \{ -J,\ldots, J\}$ for which the operator $S_z$ is diagonal. Inserting the explicit expression~\eqref{eq:repcoh} in
	$$
	\langle k | H^\prime | k^\prime  \rangle = 	\frac{2J+1}{4\pi} \int d\Omega \, N f \Big(\frac{2J}{N}  \mathbf{e}_z(\Omega) \Big) \, \langle k,J  \big| \Omega, J \rangle \langle \Omega, J \big| k^\prime, J \rangle .
	$$
	the $ \varphi $-integration of the spherical angle $ \Omega = (\theta,\varphi)$ immediately yield $ \langle k | H^\prime | k^\prime \rangle =0 $ if $k \neq k^\prime$. 
		
	It remains to control the diagonal elements of $H^\prime$ for which we fix $ k =k^\prime$ in the above integral. By a standard Taylor approximation we have 
	\[ \left|  f\Big(\frac{2J}{N}  \mathbf{e}_z(\Omega) \Big) - f\Big(\frac{2k}{N}   \Big) - \frac{2J}{N} f^\prime(2k/N)(\mathbf{e}_z(\Omega)-k/J) \right| \leq \frac{\|f^{\prime\prime}\|_{\infty} }{2} \left( \frac{2J}{N} \right)^2 (\mathbf{e}_z(\Omega)-k/J)^2 \]
	In particular we have 
	\[ \begin{split} | \langle k | H^\prime - H| k \rangle &\leq 2J \|f^{\prime}\|_{\infty} \left|   \frac{2J+1}{4\pi} \int d\Omega \,  \, (\mathbf{e}_z(\Omega)-k/J) |\langle \Omega, J \big|  k \rangle|^2  \right| \\ &+
	\frac{\|f^{\prime\prime}\|_{\infty} N }{2} \left( \frac{2J}{N} \right)^2 \frac{2J+1}{4\pi} \int d\Omega \,  \, (\mathbf{e}_z(\Omega)-k/J)^2  |\langle \Omega, J \big|  k \rangle|^2. 
 \end{split} \]
To estimate the right side, we make use of the explicit operator representations~\eqref{eq:symbolS} for $ S_z $ and~\eqref{eq:symbolS32}
which immediately yield
\[ | \langle k | H^\prime - H| k \rangle \leq C(\|f^{\prime}\|_{\infty}+ \|f^{\prime\prime}\|_{\infty}), \]
with some numerical constant $C$.

\end{proof}
Our second ingredient is an algebraic result, which allows to write any homogeneous polynomial of degree $d$ as sum of linear forms to the power $d$:

\begin{lemma}\label{lem:polalg}
Let $Q_{\text{hom}}^d(x_1,\ldots,x_k)$ be the real vector space of homogeneous polynomials of degree $d$ in the $k$ variables $x_1,\ldots x_k$.
Then, 
\[ Q_{\text{hom}}^d(x_1,\ldots,x_k) = \text{span} \{ (\alpha_1 x_1+ \cdots \alpha_k x_k)^d \, | \, \alpha_1,\ldots,\alpha_d \in \rr \} \]
\end{lemma}
We remark that this a commutative result in the sense that we distinguish between the order of variables. For example, the polynomials $x_1 x_2$ and $x_2 x_1$ are considered to be the same.

\begin{proof}
	It is clear that 
	\[ W^d(x_1,\ldots,x_k)  \coloneqq \text{span} \{ (\alpha_1 x_1+ \cdots \alpha_k x_k)^d \, | \, \alpha_1,\ldots,\alpha_d \in \rr \}  \]
is a closed subspace of $Q_{\text{hom}}^d(x_1,\ldots,x_k)$.
Let $f_{\alpha} \coloneqq (\alpha_1 x_1+ \cdots \alpha_k x_k)^d$.
Note that $f_{\alpha} - f_{\alpha^\prime} \in  W^d(x_1,\ldots,x_k)$ and since $W^d(x_1,\ldots,x_k)$ is closed, we see that 
\[ \partial_{\alpha_i}  f_{\alpha} \in W^d(x_1,\ldots,x_k).  \]
Similarly, partial derivatives of higher order are element of $W^d(x_1,\ldots,x_k)$; in particular the $d$-th order derivatives which agree with the monomials of degree $d$. Thus, $W^d(x_1,\ldots,x_k)$ contains all monomials of degree $d$ and, hence, coincides with $Q_{\text{hom}}^d(x_1,\ldots,x_k)$.
\end{proof}

We are finally ready to spell out the 
\begin{proof}[Proof of Proposition~\ref{prop:Druff1}]
Due to linearity and Lemma~\ref{lem:polalg} it enough to prove the assertion for an operator of the form 
$ H = (   \mathbf{v}\cdot    \mathbf{S})^d $, 
where $\mathbf{v}$ is some unit vector in $\rr^3$, in which case 
the claim follows from Lemma~\ref{lem:assz}.
\end{proof}

\subsection{Semiclassics for the free energy}

\begin{proposition}\label{prop:free2}
For a Hamiltonian $ H $ of the form~\eqref{eq:Hblock} on $ \mathcal{H}_N $ with a regular symbol $ h \in C^2$ satisfying~\textbf{A0}, the pressure at inverse temperature $ \beta > 0 $ is:
\begin{equation}\label{eq:freee2}
p(\beta) := \lim_{N\to \infty} N^{-1} \ln \tr \exp\left( -\beta H  \right) = \max_{r \in [0,1]} \left\{ I(r) - \beta \min_{\Omega \in S^2}  h\left(r\mathbf{e}(\Omega) \right)  \right\} .
\end{equation}
\end{proposition}
\begin{proof} 
We first use the block decomposition~\eqref{eq:Hblock} in the trace to write
$$
 \tr \exp\left( -\beta H  \right) = \sum_{J=\frac{N}{2} - \lfloor \frac{N}{2}\rfloor}^{N/2} \sum_{\alpha=1}^{M_{N,J}}  \tr_{\mathbb{C}^{2J+1}} \exp\left( -\beta H_{J,\alpha}   \right)
$$
By~\eqref{ass:symbol} the logarithm of the trace in the right side bounded according to 
$$
\sup_{J,\alpha} \left| \ln \tr_{\mathbb{C}^{2J+1}} \exp\left( -\beta H_{J,\alpha}   \right) -  \ln \tr_{\mathbb{C}^{2J+1}} \exp\left( -\beta \frac{2J+1}{4\pi} \int  N h\Big(\frac{2J}{N}  \mathbf{e}(\Omega) \Big) \,  \big| \Omega, J \rangle \langle \Omega, J \big| \, d\Omega \right) \right|  = O(1) .
$$
We then use the Berezin-Lieb inequalities~\eqref{eq:BerezinLieb} and Proposition~\ref{prop:Druff2}, which ensures that the upper and lower symbol agree up to order one, to replace the last trace by an integral, 
$$
\sup_{J,\alpha} \left|  \ln  \tr_{\mathbb{C}^{2J+1}} \exp\left( -\beta H_{J,\alpha}   \right) -  \ln \frac{2J+1}{4\pi} \int \exp\left(-\beta N  h\Big(\frac{2J}{N}  \mathbf{e}(\Omega) \Big)\right) d\Omega \right| = O(1) .
$$
In order to use a standard Laplace evaluation for the above integral in the limit $ N \to \infty $, we set $ J = r N/2 $ with fixed $ r \in (0,1] $. Since $ h \in C^2 $, we arrive at
$$
\frac{1}{N} \ln  \frac{2J+1}{4\pi} \int \exp\left(-\beta N  h\Big(\frac{2J}{N}  \mathbf{e}(\Omega) \Big)\right) d\Omega = - \beta  \min_{\Omega} h\left(r  \mathbf{e}(\Omega) \right) + o(1) . 
$$
The assertion then follows from the known asymptotics of the binomial coefficient in~\eqref{def:blockH} for $ M_{N,J} $, i.e.
$
\frac{1}{N} \ln M_{N,J} = I(r) + o(1) $. 
\end{proof}

\minisec{Acknowledgments}
We thank the referees for their suggestions to improve the manuscript. This work was supported by the Deutsche Forschungsgemeinschaft (DFG) TRR 352 Project-ID 470903074.

\bibliographystyle{plain}

\end{document}